\documentclass[12pt]{iopart}
\usepackage{graphicx}
\usepackage{iopams}
\usepackage{amsthm}

\newtheorem{lemma}{Lemma}
\newtheorem{theorem}{Theorem}
\newtheorem{corollary}{Corollary}

\usepackage{braket}
\usepackage{slashed}

\begin{document}
\title[Variational Principle
and Group Theory
for Bifurcation of Figure-eights]{Variational Principle of Action
and Group Theory
for Bifurcation of Figure-eight solutions}
\author{
Toshiaki Fujiwara$^1$,
Hiroshi Fukuda$^1$
and Hiroshi Ozaki$^2$}

\address{$^1$ College of Liberal Arts and Sciences Kitasato University, 1-15-1 Kitasato, Sagamihara, Kanagawa 252-0329, Japan}
\address{$^2$ Laboratory of general education for science and technology, Faculty of Science, Tokai University, 4-1-1 Kita-Kaname, Hiratsuka, Kanagawa, 259-1292, Japan}
\ead{fujiwara@kitasato-u.ac.jp,
fukuda@kitasato-u.ac.jp
and ozaki@tokai-u.jp}

\begin{abstract}
Figure-eight solutions are solutions to planar equal mass three-body problem
under homogeneous or inhomogeneous potentials.
They are known to be invariant under the transformation group $D_6$:
the dihedral group of regular hexagons.
Numerical investigation
shows that each figure-eight solution has 
some bifurcation points.
Six bifurcation patterns are known
with respect to the symmetry of the bifurcated solution. 

In this paper we will show the followings.
The variational principle of action
and group theory
show that the bifurcations of every figure-eight solution
are determined by the irreducible representations
of $D_6$.
Each irreducible representation has one to one correspondence
to each bifurcation.
This explains numerically observed six bifurcation patterns.
In general, in Lagrangian mechanics,
bifurcations of a periodic solution
are determined by irreducible representations of
the transformation group
that leave this solution invariant.
\end{abstract}
\vspace{2pc}

\submitto{\jpa}

\section{Introduction}
The aim of this paper is
to give a theoretical explanation for bifurcations of  figure-eight solutions
using the variational principle of action and group theory.

The plan of this paper is the following:
Section \ref{secFigureEight} and \ref{SymGroupOfFigureEight}
are devoted to introduce  figure-eight solutions
and their symmetry $D_6$.
Section \ref{secBifurcations} gives
a short history for investigation of bifurcations of figure-eight solutions.
In section \ref{secVariationalPrinciple},
we will give an application of
the variational principle of action and group theory for bifurcation.
Then, in \sref{bifurcationOfD6},
it will be shown that how the dihedral group $D_6$
determines bifurcation patterns.
In \sref{bifurcationsFigureEight}, 
interpretations of these bifurcation patterns for  figure-eight solutions
are given.
This interpretation explains numerical results of bifurcations
of figure-eight solutions.
Summary and discussions are given in \sref{summaryDiscussions}.

\subsection{Figure-eight solutions}
\label{secFigureEight}
A figure-eight solution is a periodic solution
to planar equal mass three-body problem.
In this solution,
three masses chase each other around one eight-shaped orbit
with equal time delay:
\begin{equation}
\label{eqChoreographic}
r_k(t)=r_0(t+k T/3),\, r_0(t+T)=r_0(t).
\end{equation}
Here,  $r_k(t) \in \mathbb{R}^2$, $k=0,1,2$ represents position of particle $k$
and $T$ is the period.
The first figure-eight solution 
was numerically found by C.~Moore in 1993 \cite{Moore}
under the homogeneous potential (opposite sign of the potential energy)
\begin{equation}
U=U_h=\frac{1}{a}\sum_{i<j}\frac{1}{r_{i j}^a}
\mbox{ for } a>-2,
\end{equation}
where $r_{i j}=|r_i-r_j|$.
The total Lagrangian is
\begin{equation}
\label{Lagrangian}
L=\frac{1}{2}\sum_{k=0,1,2} \left|\frac{dr_k}{d t}\right|^2+U.
\end{equation}
A.~Chenciner and R.~Montgomery in 2000 \cite{ChencinerMontgomery}
proved its existence rigorously. 
Sometimes, this figure-eight solution is called ``the figure-eight'' solution.

Soon after that,
C.~Sim\'o found many planar $N$-body solutions \cite{SimoChoreography}
in which $N$ bodies chase each other around a single closed loop  with equal time delay.
He called such solutions ``choreographies''.
He also found a non-choreographic orbit near the figure-eight one,
that is now called  Sim\'{o}'s H solution  \cite{SimoDynamical}.

In 2004,
L.~Sbano found a figure-eight solution under the Lennard-Jones potential
\begin{equation}
U=U_{LJ}=\sum_{i<j} \left(\frac{1}{r_{i j}^6} -\frac{1}{r_{i j}^{12}}\right)
\end{equation}
for sufficiently large period $T$ \cite{Sbano2005}.
At the same time, he proved that at least one more figure-eight solution exists at the same period
\cite{Sbano2005}.
See also Sbano and J.~Southall \cite{SbanoSouthall} in 2010.
However, no one can find the predicted solution until 2016.
This is because 
figure-eight 
solutions founded at that time, under homogeneous or Lennard-Jones,
are local minimizer of the action,
while the predicted another figure-eight solution should be a saddle
\cite{Sbano2005, SbanoSouthall}.

In 2016,
one of the present authors, H.~F, 
developed a method to search figure-eight solutions
inspired by a method by 
M.~\v{S}uvakov and V.~Dmitra\v{s}inovi\'{c} \cite{Suvakov2013,Suvakov2013NumericalSearch}.
These methods are free from the action minimizing process.
He immediately applied his method to the Lennard-Jones potential
and found many figure-eight solutions. 
He named them $\alpha_\pm,\ \beta_\pm,\ \dots,\, \epsilon_\pm$
\cite{Fukuda2017}.
Here, $\pm$ solutions are connected by fold bifurcation with period $T$
and
$+$ solution has greater action than $-$ solution at the same $T$.
The figure-eight solution found by Sbano is $\alpha_-$, a local minimizer of action,
and $\alpha_+$ is the predicted saddle.

\subsection{Symmetry group for  figure-eight solutions: $D_6$}
\label{SymGroupOfFigureEight}
A figure-eight solution has many symmetries.
Let $\sigma$ and $\tau$ be operators that change the suffix of particles.
For $q=(r_0,r_1,r_2) \in \mathbb{R}^6$,
\begin{equation}
\sigma(r_0,r_1,r_2)=(r_1,r_2,r_0),
\,
\tau(r_0,r_1,r_2)=(r_0,r_2,r_1).
\end{equation}
And let $\mu_x$ be the operator representing inversion in the $Y$-axis.
For $r_k=(r_{k x}, r_{k y})$,
\begin{equation}
\mu_x (r_{k x}, r_{k y})=(-r_{k x}, r_{k y}).
\end{equation}
Moreover,
let $\mathcal{R}^{1/6}$ for time shift of $T/6$,
and $\Theta$ for time inversion,
\begin{equation}
\mathcal{R}^{1/6}q(t)=q(t+T/6),\,
\Theta q(t)=q(-t).
\end{equation}
Finally using these operators, let $\mathcal{B}$ and $\mathcal{S}$ be
\begin{equation}
\mathcal{B}=\sigma \mu_x \mathcal{R}^{1/6},
\mbox{ and }
\mathcal{S}=-\tau \Theta.
\end{equation}
Then a figure-eight solution $q_o$ is invariant under these operations,
\begin{equation}
\label{invarianceOfQo}
\mathcal{B}q_o=\mathcal{S}q_o=q_o.
\end{equation}
Inversely, if a solution is invariant for $\mathcal{B}$ and $\mathcal{S}$,
the solution is called a figure-eight solution.
Especially, the invariance under 
\begin{equation}
\mathcal{C}=\mathcal{B}^2=\sigma^{-1}\mathcal{R}^{1/3}
\end{equation}
represents the choreographic symmetry in \eref{eqChoreographic}.

By the definitions,
the operators $\mathcal{B}$ and $\mathcal{S}$ satisfy
\begin{equation}
\label{eqDefD6}
\mathcal{B}^6=\mathcal{S}^2=1
\mbox{ and }
\mathcal{B}\mathcal{S}=\mathcal{S}\mathcal{B}^{-1}.
\end{equation}
The group generated by $\mathcal{B}$ and $\mathcal{S}$ is the dihedral group $D_6$:
\begin{equation}
D_6=\{1,\mathcal{B},\mathcal{B}^2,\dots,\mathcal{B}^5,
	\mathcal{S},\mathcal{S}\mathcal{B},
	\mathcal{S}\mathcal{B}^2,\dots,
	\mathcal{S}\mathcal{B}^5
	\}.
\end{equation}
For this group, $\{1, \mathcal{B}^3\}$ is the centre,
and the commutator subgroup is $\{1, \mathcal{C},\mathcal{C}^2\}$.
In the following, we will write
\begin{equation}
\mathcal{M}=\mathcal{B}^3.
\end{equation}

We will use the invariance of the action $S$ under $D_6$,
\begin{eqnarray}
S[\mathcal{B}q]=S[\mathcal{S}q]=S[q]\label{invarianceOfAction}
\end{eqnarray}
for any periodic function $q(t)$ with period $T$.
This is true
if three masses are equal and the potential has the form
$U=\sum_{i<j} u(r_{i j})$,
because the action is invariant under the transformation 
of time shift, time inversion,
spacial inversion
and exchange of particles.

In this paper,
a group $G$ is called  a symmetry group for $q_o$ and the action $S$
if $g q_o=q_o$ and $S[g q]=S[q]$
for any periodic function $q$ and every $g\in G$.
If $G$ is a symmetry group,
then a subgroup of $G$ is also a symmetry group.
The above $D_6$ is a symmetry group for figure-eight solutions and for the action.
Linear operators will be written by calligraphic style such as
$\mathcal{B},\, \mathcal{S}$, etc.,
their eigenvalues by $'$ such as $\mathcal{B}',\, \mathcal{S}'$, etc.,
and their representation by tilde
such as $\tilde\mathcal{B},\, \tilde\mathcal{S}$, etc.


\subsection{Bifurcations of  figure-eight solutions}
\label{secBifurcations}
In 2002,
J.~Gal\'an,
F.~J.~Mu\~noz-Almaraz,
E.~Freire,
E.~Doedel,
and 
A.~Vanderbauwhede
pointed out that
the figure-eight solution and Sim\'o's H solution
are connected by a fold bifurcation
by taking the mass of one particle as a parameter \cite{PhysRevLett.88.241101}.
Namely,
the total Lagrangian they use is
\begin{equation}
\label{lagrangianGalan}
L=\frac{1}{2}\sum_k m_k\left|\frac{dr_k}{d t}\right|^2+\sum_{i<j}\frac{m_i m_j}{r_{i j}}
\end{equation}
with $m_1=m_2=1$ and $m_0$ is the parameter.
This was the first observation to connect the figure-eight and Sim\'o's H solution.
See also 
Mu\~noz-Almaraz,
Gal\'an, and
Freire \cite{Munoz2004}
in 2004.

In 2004, one of the present authors, T.~F., met Vanderbauwhede
and asked him
what will happen if we take the potential $U_h$ and take $a$ as parameter.
This is because the mass parameter $m_0\ne 1$ in \eref{lagrangianGalan} breaks the $D_6$ symmetry
of the figure-eight solution down to $D_2$.
As a result, the solution at $m_0\ne 1$ is not choreographic.
While, 
as shown by Moore \cite{Moore},
the potential $U_h$ keeps $D_6$ symmetry,
therefore the solution keeps choreographic figure-eight shape.
His team
 immediately calculated and found bifurcations
at $a=0.9966$ and $1.3424$ \cite{Munoz2018}.
The bifurcation at $a=0.9966$ yields Sim\'o's H solution.
T.~F. received their notes \cite{Munoz2018} from
Mun\~oz-Almaraz
in June 2005.

In 2018 and 2019,
the present authors investigated bifurcations of figure-eight solutions
under the homogeneous potential $U_h$ with parameter $a$ and 
of $\alpha_\pm$ solution under $U_{LJ}$ with period $T$
using Morse index of the action.
Here Morse index stands for number of negative eigenvalues of the second derivative of the action.
We confirmed the results of Mun\~oz-Almaraz et al.~\cite{Munoz2018} in $U_h$,
and found many bifurcations for $\alpha_\pm$ \cite{Fukuda2018, Fukuda2019}
in $U_{LJ}$.

At that time, by  numerical calculations,
we noticed that  bifurcations occur when Morse index changed,
and also noticed that the symmetry of eigenfunction that 
is responsible to  change of Morse index
determines the symmetry of bifurcated solutions
and bifurcation patterns.
The correspondence between symmetry of eigenfunction
and symmetry of bifurcated function and bifurcation pattern
is one-to-one.
We consider the reason and give an answer.
These bifurcations are explained by the variational principle of action
and group theory.
We will show that bifurcation patterns correspond to irreducible representations
of $D_6$.

Group theoretical method of bifurcation
is well known
in condensed matter physics and particle physics.
We apply this method to bifurcations for periodic solutions.
As shown in this paper,
we can understand bifurcations of a periodic solution
as a zero point of first derivative of the action.
This will give an alternative point of view to investigate
the bifurcations,
other than that based on  Poincaré map or Floquet matrix.


\section{Variational principle of action for bifurcation}
\label{secVariationalPrinciple}
To show the basic idea for the variational principle of action for bifurcation,
let us consider a function $f(x)$ whose derivatives are $f'(x),\, f''(x),\, f'''(x),\dots$.
The Taylor series reveals the function $f$ near an arbitrary point $x_0$,
\begin{equation}
f(x_0+r)=f(x_0)+f'(x_0)r+\frac{f''(x_0)}{2}r^2+\frac{f'''(x_0)}{3!}r^3+\dots.
\end{equation}
A stationary point is a point that satisfies $f'(x_0)=0$.
If two stationary points $x_o$ and $x_b=x_o+r$ exist and $r \to 0$,
we have $f''(x_o) \to 0$, because
\begin{equation}
0=\frac{f'(x_o+r)-f'(x_o)}{r} \to f''(x_o).
\end{equation}
Inversely, for a stationary point $x_o$,
if $\kappa=f''(x_o)$ crosses zero and $f'''(x_o)\ne 0$,
the first derivative of the Taylor series at $x_0$
is given by
\begin{equation}
\partial_r f(x_o+r)
=\left(\kappa+\frac{f'''(x_o)}{2}r + O(r^2)\right)r.
\end{equation}
This has two zeros:
$r_o=0$ and  
\begin{equation}
r_b=-2\kappa/f'''(x_o)+O(\kappa^2),
\end{equation}
corresponding to  stationary points $x_o$ and $x_b$.
The value of $f(x_b)$ and $f''(x_b)$ are
\begin{eqnarray}
f(x_b)=f(x_o+r_b)=f(x_o)+\frac{2\kappa^2}{3(f'''(x_o))^2}+O(\kappa^3),\\
f''(x_b)=f''(x_o+r_b)=-\kappa+O(\kappa^2).
\end{eqnarray}
Thus,
higher derivatives of function $f$ at a stationary point $x_o$
can tell us the existence of another stationary point
and the values of $x_b$, $f(x_b)$ and $f''(x_b)$ in a series of $\kappa$.

By the variational principle of action,
a stationary point of action is a solution of the equations of motion.
Almost the same procedure for the action 
makes a theory of bifurcation for periodic solutions.
In \sref{secDerivativesOfAction},
the second derivative of the action at a solution $q_o$
will define the
Hessian operator $\mathcal{H}(q_o)$ of the action.
In \sref{secNecessaryCondition}, the necessary condition for bifurcation will be given.
In \sref{secExpansionOfAction},
the action will be expanded by the eigenfunctions 
of $\mathcal{H}(q_o)$.
Here, the eigenfunctions are characterized by 
the irreducible representations of 
the symmetry group $G$ for $q_o$ and $S$. 
Then the action is a function of coefficients
of eigenfunctions that are infinitely many.
A method, that is called Lyapunov-Schmidt reduction,
will be used to reduce the number of variables
to the degeneracy number $d$ of an 
eigenvalue.
Lyapunov-Schmidt reduction will be shown in 
\sref{secLyapunovSchmidt}.
It will be shown that
the necessary condition is also a sufficient condition
for $d=1$ and $d=2$
in sections \ref{secNonDegeneratedCase} and \ref{secDegeneratedCase}
respectively.
In \sref{secProjectionOperator},
a quite useful method 
to predict the existence of bifurcation 
and to predict the symmetry of bifurcated solution
utilizing a projection operator is shown.
Finally, in \sref{secInheritance}, two equalities are listed.
One is useful to calculate higher derivatives of action,
and the other is a symmetry of Lyapunov-Schmidt reduced action.

The methods described here,
utilizing irreducible representation,
Lyapunov-Schmidt reduction
and projection operator,
are well known
in condensed matter physics, particle physics, etc.
to investigate
symmetry breaking and bifurcations in 
material, universe, etc. 
\cite{sattinger,golubitsky,golubitskyII,IkedaMurota}.
However, as far as we know, 
A.~Chenciner, J.~F\'{e}joz and R.~Montgomery \cite{RotatingEight} 
gave only one application of these methods to bifurcations of a periodic solution. 
In this section, let us shortly summarize these methods in case the readers do not 
well versed in. 

Although the aim of this paper is to describe bifurcations of
figure-eight solutions,
let us treat bifurcations of a general periodic solution $q_o$
with 
symmetry 
group $D_n$,
instead of  figure-eight solutions and $D_6$.


\subsection{Higher derivatives of action and definition of Hessian}
\label{secDerivativesOfAction}
Let us consider a system
that described by generalized coordinates $q_i$, $i=1,2,\dots,N$,
with action
\begin{equation}
S[q]=\int_0^T d t \left(
		\frac{1}{2}\sum_{i=1,2,\dots,N}m_i\left(\frac{dq_i}{d t}\right)^2+U(q)
	\right).
\end{equation}
We consider the function space with period $T$,
\begin{equation}
q(t+T)=q(t).
\end{equation}
The derivatives of the action around a function $q$ are
\begin{equation}
S[q+\delta q]
=S[q]+\delta S[q]+\frac{1}{2}\delta^2 S[q]+\frac{1}{3!}\delta^3 S[q]+\frac{1}{4!}\delta^4 S[q]+\dots.
\end{equation}
Here, the
variation $\delta q(t)$ is also periodic with the same period $T$.
The derivatives are,
\begin{eqnarray}
\delta S[q]
=\int_0^T d t \sum_i \delta q_i\left(
		-m_i\frac{d^2 q_i}{d t^2}+\frac{\partial U}{\partial q_i}
	\right),\\
\delta^2 S[q]
=\int_0^T d t \sum_i \delta q_i\left(
		-\delta_{i j}m_i\frac{d^2}{d t^2}+\frac{\partial^2 U}{\partial q_i\partial q_j}
	\right)
	\delta q_j,\\
\mbox{and }\delta^n S[q]=\braket{(\delta q)^n}
\mbox{ for } n\ge 3.
\dots.
\end{eqnarray}
Here, $\delta_{ij}$ is the Kronecker delta,
and abbreviated notations $\braket{(\delta q)^n}$ are
\begin{eqnarray*}
\braket{(\delta q)^3}
=\int_0^T d t \sum_{ijk}\frac{\partial^3 U}{\partial q_i \partial q_j \partial q_k}
		\delta q_i \delta q_j \delta q_k,\\
\braket{(\delta q)^4}
=\int_0^T d t \sum_{ijk\ell}\frac{\partial^4 U}{\partial q_i \partial q_j \partial q_k \partial q_\ell}
		\delta q_i \delta q_j \delta q_k \delta q_\ell,
\mbox{ and so on}.
\end{eqnarray*}
In general, we will use  abbreviated notations $\braket{fg\dots h}$ for
\begin{equation}
\braket{fg\dots h}
=\int_0^T d t \sum_{ij\dots\ell}\left(\frac{\partial}{\partial q_i}\frac{\partial}{\partial q_j}
			\dots\frac{\partial U}{\partial q_\ell}\right)
		f_i g_j \dots h_\ell.
\end{equation}
The operator in the second derivative defines the Hessian operator $\mathcal{H}$:
\begin{equation}
\mathcal{H}_{i j}=
-\delta_{ij}m_i\frac{d^2}{d t^2}+\frac{\partial^2 U}{\partial q_i\partial q_j}.
\end{equation}
By the variational principle,
a stationary point that satisfies $\delta S[q]=0$ is a solution of the equations of motion:
\begin{equation}
\label{eqEquationOfMotion}
m_i \frac{d^2 q_i}{d t^2}=\frac{\partial U}{\partial q_i}.
\end{equation}

\subsection{Necessary condition for bifurcation}
\label{secNecessaryCondition}
What will happen at a bifurcation point?
Consider a region of a parameter $\xi$ where 
both original solution $q_o$ and a bifurcated solution $q_b$ exist
and $q_b \to q_o$ for $\xi\to \xi_0$.
The point $\xi_0$ is a bifurcation point.
Let $q_b=q_o+R\Phi$ with
\begin{equation}
||\Phi||^2=\int_0^T d t \Phi(t)^2=1.
\end{equation}
Since both $q_b$ and $q_o$ satisfy the equations of motion \eref{eqEquationOfMotion},
the difference satisfies
\begin{equation*}
R m_i \frac{d^2 \Phi_i}{d t^2}=R\frac{\partial^2 U}{\partial q_i \partial q_j}\Phi_j + O(R^2).
\end{equation*}
Namely,
\begin{equation}
\label{eqNecessaryCondition}
\mathcal{H}_{i j}\Phi_j = O(R) \to 0 \mbox{ for } \xi \to \xi_0.
\end{equation}
So, one of the eigenvalue of $\mathcal{H}$ tends to zero for $\xi$ tends to a bifurcation point.
This is a necessary condition for bifurcation.

To get more precise condition,
let us expand $q_b-q_o$ by eigenfunctions of $\mathcal{H}$:
\begin{eqnarray}
q_b-q_o=r\phi + \sum_\alpha r\epsilon_\alpha \psi_\alpha,\\
\mathcal{H}\phi=\kappa\phi,
\,
\mathcal{H}\psi_\alpha=\lambda_\alpha\psi_\alpha,
\end{eqnarray}
where $\kappa \to 0$ and $\lambda_\alpha$ remain finite for $\xi \to \xi_0$.
The functions $\phi$ and $\psi_\alpha$ are normalized orthogonal functions.
Following the same arguments from \eref{eqEquationOfMotion} to
\eref{eqNecessaryCondition}, we get
\begin{equation}
r\mathcal{H}\left(
	\phi + \sum_\alpha \epsilon_\alpha \psi_\alpha
	\right)
	=r\kappa\phi + \sum_\alpha \lambda_\alpha r\epsilon_\alpha \psi_\alpha
	=O(r^2).
\end{equation}
Dividing it by $r\ne 0$,
\begin{equation}
\kappa\phi + \sum_\alpha \lambda_\alpha \epsilon_\alpha \psi_\alpha
=O(r).
\end{equation}
Therefore,
\begin{equation}
\kappa=O(r)
\mbox{ and }
\epsilon_\alpha = O(r).
\end{equation}
So, the function $\phi$ contributes the difference $q_b-q_o$ in $O(r)$,
whereas $\psi_\alpha$ in $r\epsilon_\alpha=O(r^2)$.

\subsection{Expression for higher derivatives in terms of eigenfunctions of $\mathcal{H}$}
\label{secExpansionOfAction}
Now, we expand an arbitrary variation $\delta q$ in terms of eigenfunctions of $\mathcal{H}$:
\begin{equation}
\label{eqExpansionOfDeltaQ}
\delta q
	= r \phi + \sum_\alpha r\epsilon_\alpha \psi_\alpha.
\end{equation}
In this expansion,
we exclude eigenfunctions that always belong to zero eigenvalue independent from the parameter.
Because these eigenfunctions are connected to
conservation laws,
they are irrelevant to bifurcation.
For example, eigenfunction $d q_o/d t$ belongs to zero eigenvalue of $\mathcal{H}$
that is connected to the energy conservation law.
Actually, 
$q_o(t)+\epsilon d q_o/d t +O(\epsilon^2)$
is just a time shift of $q_o(t+\epsilon)$.
So, the eigenfunction $d q_o/d t$ is irrelevant to bifurcation.
Therefore, all of the eigenfunctions in \eref{eqExpansionOfDeltaQ} should be understood to 
have non-zero eigenvalue in general.
Then, for $\xi\to \xi_0$,
only $\kappa$, the eigenvalue of $\phi$, tends to 0.
In the following,
we consider a sufficiently small region of parameter $\xi$
where only $\kappa \to 0$ for $\xi \to \xi_0$ and the absolute value of  all $\lambda_\alpha$ 
are greater than a positive value.

The eigenvalue $\kappa$ may have  degeneracy $d=2$.
(The irreducible representations of $D_n$ has 
degeneracy at most $d=2$.)
In this case, there are two linearly independent eigenfunctions
for $\mathcal{H}\phi_\gamma = \kappa \phi_\gamma$, $\gamma=1,2$.
In this case,
we use polar coordinates $(r, \theta)$. 
Namely, $r\phi(\theta)=r(\cos(\theta)\phi_1+\sin(\theta)\phi_2)$ if $d=2$.
If $d=1$, $\phi(\theta)$ should be understood as $\phi$.

Now consider the expansion for $S[q_o+\delta q]-S[q_o]$ for \eref{eqExpansionOfDeltaQ}
for $q_o$ that is a solution of the equation of motion.
By the variational principle of action, the first derivative vanishes: $\delta S[q_o]=0$.
The higher derivatives are
\begin{equation}
\label{eqSecondOrder}
\delta^2 S[q_o]
=r^2\left(
	\kappa
	+\sum_\alpha \lambda_\alpha \epsilon_\alpha^2
\right),
\end{equation}
\begin{equation}
\label{eqHigherOrder}
\delta^n S[q_o]
=r^n
	\Braket{\Big(\phi(\theta)+\sum_\alpha \epsilon_\alpha \psi_\alpha\Big)^n}
\mbox{ for } n \ge 3.
\end{equation}
Thus,
using the expressions \eref{eqSecondOrder} and \eref{eqHigherOrder},
the difference of the action is written
by the infinitely many variables
$r, \theta, \epsilon$:
\begin{equation}
S(r,\theta,\epsilon)
=S[q_o+\delta q]-S[q_o]
=\frac{1}{2}\delta^2 S[q_o]+\sum_{n\ge 3} \frac{1}{n!}\delta^n S[q_o].
\end{equation}

\subsection{Lyapunov-Schmidt reduction}
\label{secLyapunovSchmidt}
To find a bifurcated solution $q_b$,
we solve the stationary conditions
\begin{equation}
\partial_r S(r,\theta,\epsilon)=\partial_\theta S(r,\theta,\epsilon)=\partial_\epsilon S(r,\theta,\epsilon)=0.
\end{equation}
Instead of solving these equations at once,
we first solve the equations for $\epsilon$ for arbitrary $r,\, \theta$.
Namely,
\begin{eqnarray}
\fl
\frac{\partial}{\partial \epsilon_\alpha}S(r,\theta,\epsilon)
=r^2\lambda_\alpha \epsilon_\alpha
	+\sum_{n\ge 3}\frac{r^n}{(n-1)!}
		\Braket{
			\Big(\phi(\theta)+\sum_\beta \epsilon_\beta \psi_\beta\Big)^{n-1}
			\psi_\alpha
		}
=0.
\end{eqnarray}
Dividing it by $r^2\lambda_\alpha$, we obtain
a recursive equation for $\epsilon$:
\begin{equation}
\label{recEqForEpsilon}
\epsilon_\alpha
=-\frac{1}{\lambda_\alpha}
	\sum_{n\ge 3}\frac{r^{n-2}}{(n-1)!}
		\Braket{
			\Big(\phi(\theta)+\sum_\beta \epsilon_\beta \psi_\beta\Big)^{n-1}
			\psi_\alpha
		}.
\end{equation}
Solving this equation, we get $\epsilon_\alpha$ in a series of $r$ uniquely:
\begin{equation}
\epsilon_\alpha(r,\theta)
=-\frac{r}{2\lambda_\alpha}\Braket{\phi(\theta)^2 \psi_\alpha}
	+O(r^2).
\end{equation}
Substituting this solution into $S(r,\theta,\epsilon)$, we obtain a reduced action
\begin{equation}
S_{LS}(r,\theta)=S(r,\theta,\epsilon(r,\theta)).
\end{equation}
This procedure is known as Lyapunov-Schmidt reduction.
The reduced action is a function of 
$r, \theta$:
\begin{equation}
\label{SLS}
S_{LS}(r,\theta)=\frac{\kappa}{2}r^2+\sum_{n\ge 3}\frac{A_n(\theta)}{n!}r^n.
\end{equation}
Explicit expressions for lower $A_n$ are
\begin{eqnarray}
\label{A3A4andA5}
\frac{A_3(\theta)}{3!}
=\frac{1}{3!}\braket{\phi(\theta)^3}&,\label{A3}\\
\frac{A_4(\theta)}{4!}
=\frac{1}{4!}\braket{\phi(\theta)^4}&
			-\frac{1}{(2!)^3}\sum_\alpha
				\frac{\Braket{\phi(\theta)^2 \psi_\alpha}^2}{\lambda_\alpha}
				,\label{A4}\\
\frac{A_5(\theta)}{5!}
=\frac{1}{5!}\braket{\phi(\theta)^5}&
-\frac{1}{2!3!}\sum \frac{\braket{\phi(\theta)^2\psi_\alpha}\braket{\phi^3\psi_\alpha}}
					{\lambda_\alpha}\nonumber\\
&+\frac{1}{(2!)^3}\sum \frac{
					\braket{\phi(\theta)^2\psi_\alpha}
					\braket{\phi\psi_\alpha\psi_\beta}
					\braket{\phi(\theta)^2\psi_\beta}
					}
					{\lambda_\alpha \lambda_\beta},\label{A5}
\end{eqnarray}
and so on.

\subsection{Bifurcations for non-degenerate case}
\label{secNonDegeneratedCase}
If the eigenvalue $\kappa$ is not degenerate,
the reduced action is a function of single variable $r$:
\begin{equation}
S_{LS}(r)=\frac{\kappa}{2}r^2+\sum_{n\ge 3}\frac{A_n}{n!}r^n.
\end{equation}
Then the first derivative 
\begin{equation}
S'_{LS}(r)=\left(\kappa+\sum_{n\ge 3}\frac{A_n}{(n-1)!}r^{n-2}\right)r
\end{equation}
has two zeros,
$r=0$ for the original solution and
\begin{equation}
\label{eqBifurcatedSolutionForR}
\kappa+\sum_{n\ge 3}\frac{A_n}{(n-1)!}r^{n-2}
=\kappa+\frac{A_3}{2}r+\frac{A_4}{3!}r^2 + \dots =0
\end{equation}
for the bifurcated solution.

If $A_3\ne 0$, the equation \eref{eqBifurcatedSolutionForR} has
the solution
\begin{equation}
\label{eqFirstOrderRb}
r_b=-\frac{2}{A_3}\kappa-\frac{4A_4}{3 A_3^3}\kappa^2+O(\kappa^3).
\end{equation}
Bifurcated solution exists in both negative and positive side of $\kappa$.
For this case, we have
\begin{eqnarray}
S_{LS}(r_b)=S[q_b]-S[q_o]
	=\frac{2}{3 A_3^2}\kappa^3
		+\frac{2A_4}{3 A_3^4}\kappa^4
		+O(\kappa^5),\label{eqFirstOrderSLS}\\
S''_{LS}(r_b)
	=-\kappa
		+\frac{2A_4}{3 A_3^2}\kappa^2
		+O(\kappa^3)\label{eqFirstOrderSLSSecond}.
\end{eqnarray}
Therefore, for sufficiently small $\kappa$,
the difference of action $S[q_b]-S[q_o]$  is  proportional to $\kappa^3$ and the coefficient must be positive.
And the second derivative must have opposite sign and the same magnitude for the original solution
for small $\kappa$.
In  \ref{EigenvaluesOfH}, we will show that $S''_{LS}(r_b)$
gives correct value of 
the eigenvalue of the Hessian $\mathcal{H}$ for the bifurcated solution to 
the order in \eref{eqFirstOrderSLSSecond}.

Although
this is the simplest bifurcation,
this case describes 
``both-side'' bifurcation or fold one
depending on the relation between the parameter $\xi$ and the eigenvalue $\kappa$:
\begin{equation}
\xi-\xi_0
=a_1 \kappa + a_2 \kappa^2 + \dots.
\end{equation}

If $a_1\ne 0$, where correspondence between $\xi$ and $\kappa$ is one to one,
\eref{eqFirstOrderRb} describes ``both-side'' bifurcation.
The bifurcated solution exists for both sides of $\xi-\xi_0$.

While if $a_1=0$ and $a_2\ne 0$,
where the correspondence between $\xi$ and $\kappa$ is one to two,
\eref{eqFirstOrderRb} describes fold bifurcation.
The bifurcated solution exists only one side of $\xi$:
$\xi-\xi_0>0$ if $a_2>0$ or $\xi-\xi_0<0$ if $a_2<0$.
From the point of view of $\xi$,
two solutions are created or annihilated at $\xi_0$.

Now, let us proceed to the case $A_3=0$ and $A_4\ne 0$.
This is order $2$ bifurcation
because $\kappa=O(r^2)$.
In general, the solution of \eref{eqBifurcatedSolutionForR} is
\begin{equation}
r_b=\pm\sqrt{\frac{-6\kappa}{A_4}}+\frac{3A_5}{4 A_4^2}\kappa+O(\kappa^{3/2}).
\end{equation}
This describes ``one-side''  bifurcation.
The bifurcated solutions exist at one side of $\kappa$.
Two bifurcated solutions emerge for 
$\kappa>0$ if $A_4<0$,
or $\kappa<0$ if $A_4>0$.
The terms $A_{2 n+1}$, $n\ge2$ breaks $r \to -r$ symmetry of $S_{LS}$.

Although these general cases will be interesting,
sometimes a symmetry makes $S_{LS}(-r)=S_{LS}(r)$.
This symmetry makes $A_{2 n+1}=0$ for all $n\ge 1$.
In this case, the solution of \eref{eqBifurcatedSolutionForR} is
\begin{equation}
\label{order2Rb}
r_b=\pm \left(\frac{-6\kappa}{A_4}\right)^{1/2}
\left(
1+\frac{3A_6}{20 A_4^2}\kappa + O(\kappa^{2})
\right).
\end{equation}
The two solutions have exactly equal $|r_b|$,
$S_{LS}$ and $S''_{LS}$,
\begin{eqnarray}
S_{LS}(r_b)=-\frac{3}{2 A_4}\kappa^2-\frac{3A_6}{10 A_4^3}\kappa^3+O(\kappa^4),\label{order2SLS}\\
S''_{LS}(r_b)=-2\kappa+\frac{3A_6}{5 A_4^2}\kappa^2+O(\kappa^3) \label{order2S''LS}.
\end{eqnarray}
Since $A_4$ and $\kappa$ have opposite signs, the sign of $S_{LS}(r_b)$ is the same as $\kappa$.

In general, if $A_{3}=A_{4}=\dots=A_{n-1}=0$ and $A_{n} \ne 0$ for odd $n$
\begin{equation}
r_b=\left(-\frac{(n-1)!}{A_n}\kappa\right)^{1/(n-2)}
\end{equation}
for both side of $\kappa$,
while for even $n$
\begin{equation}
r_b=\pm \left(-\frac{(n-1)!}{A_n}\kappa\right)^{1/(n-2)}
\end{equation}
for one side of $\kappa$: $\kappa/A_n <0$.

What will happen if $A_n=0$ for all $n$?
In this case, however,
the reduced action is exactly
$S_{LS}(r)=S(q_o)+\kappa r^2/2$.
This action behaves badly.
Consider the behaviour  of this reduced action at sufficiently large $r=M$.
For the small interval of $-1/M < \kappa < 1/M$,
the change of action is huge, since $-M/2 < \kappa r^2/2 <M/2$.
Although, we didn't find a logic to exclude this case,
this case unlikely exists. 
We simply neglect this case in this paper.
Actually,
as far as we know,
there are no symmetry that makes $A_4=0$.

As shown in this subsection, 
for non-degenerate case,
the condition $\kappa \to 0$ is a sufficient condition for bifurcation,
except for the case in the last paragraph
that unlikely exists.


\subsection{Bifurcations for degenerate case}
\label{secDegeneratedCase}
If the eigenvalue $\kappa$ has degeneracy $d=2$,
the reduced action has $\theta$ dependence:
\begin{equation}
S_{LS}(r,\theta)=\frac{\kappa}{2}r^2+\sum_{n\ge 3}\frac{A_n(\theta)}{n!}r^n.
\end{equation}
Then the condition for stationary point for $\theta$ is
\begin{equation}
\label{eqTheta}
\partial_\theta S_{LS}(r,\theta)=0.
\end{equation}
Solving this equation, and substituting the solution $\theta(r)$ into $S_{LS}(r,\theta)$
yields one variable function $S_{LS}(r,\theta(r))$.
Then the same arguments for non-degenerate case
will be used to describe the sufficient condition for bifurcation.

There may several solutions $\theta(r)$ of equation \eref{eqTheta}.
In such case, each solution yields each $S_{LS}(r,\theta(r))$.
As a result, multiple bifurcated solutions corresponds to each solution $\theta(r)$
will emerge from the original at $\kappa=0$.

Can we tell what $\theta$ will give stationary point(s)?
The information must be in the behaviour of $S_{LS}(r,\theta)$
that is determined by 
the symmetry group $G$.
In the next subsection,
we will show that
the group structure of $G$ surely determines the
direction $\theta$ for bifurcation.

\subsection{Projection operator}
\label{secProjectionOperator}
An idempotent operator defined by $\mathcal{P}^2=\mathcal{P}$ is a projection operator.
We are interested in projection operators
that leave
$q_o$ invariant: $\mathcal{P} q_o=q_o$.
Let $G$ be a  symmetry group for $q_o$
and the action.
Then average of elements of  $G$ defines a projection operator
\begin{equation}
\mathcal{P}_{G}=\frac{1}{|G|}\sum_{g \in G}g,
\end{equation}
where $|G|$ is the number of the elements in $G$.
Since $G$ is a group, $g G=G$ then $g\mathcal{P}_G=\mathcal{P}_G$
 and $\mathcal{P}_G^2=\mathcal{P}_G$.
So $\mathcal{P}_{G}$ is a projection operator
that satisfies $\mathcal{P}_{G}q_o=q_o$.
It is obvious if $g f=f$ for every $g \in G$ then $\mathcal{P}_{G}f=f$.
The inverse is also true:
if $\mathcal{P}_{G}f=f$ then 
$g f=g(\mathcal{P}_{G}f)=(g\mathcal{P}_{G})f=\mathcal{P}_{G}f=f$
for every $g \in G$.

There is another way to make projection operator.
An arbitrary $g \in G$ defines the cyclic subgroup 
\begin{equation}
\{1, g, g^2,\dots,g^{n_g-1}\}
\end{equation}
with an integer $n_g$ that is the smallest natural number 
of $g^{n_g}=1$.
This $n_g$ is called the order of the element $g$.
If there is no confusions,
let us write this projection operator $\mathcal{P}_g$ for brevity:
\begin{equation}
\mathcal{P}_g=\frac{1}{n_g}\sum_{n=0,1,2,\dots,n_g-1}g^n.
\end{equation}
For $G=D_6$,
\begin{equation}
\label{projectionOperators}
\fl
\mathcal{P_{C}}=\frac{1}{3}(1+\mathcal{C}+\mathcal{C}^2),\,
\mathcal{P_{M}}=\frac{1}{2}(1+\mathcal{M}),\,
\mathcal{P_{S}}=\frac{1}{2}(1+\mathcal{S}),
\mathcal{P_{MS}}=\frac{1}{2}(1+\mathcal{MS})
\end{equation}
and
\begin{equation}
\mathcal{P}_{D_6}=\mathcal{P_{C}}\mathcal{P_{M}}\mathcal{P_{S}}.
\end{equation}
A projection operator $\mathcal{P}$
decomposes any function $f$ into two parts
$f=\mathcal{P}f+(1-\mathcal{P})f$,
where $\mathcal{P}f$ belongs to the eigenvalue $\mathcal{P}'=1$
and $(1-\mathcal{P})f$  belongs to $\mathcal{P}'=0$.

For a symmetry group $G$ for $q_o$ and $S$, 
the projection operator $\mathcal{P}_G$ splits 
eigenfunctions of $\mathcal{H}$ 
into two sets: 
one set  that belong to $\mathcal{P}_G'=1$ 
and the other set  that belong to 
$\mathcal{P}_G'=0$. 
We write set of  eigenfunctions
$\{e_1,e_2,e_3,\dots\}
=\{\phi_1,\dots,\phi_d,\psi_1,\psi_2,\psi_2,\dots\}$
and
\begin{equation}
r\phi(\theta)+r\sum \epsilon_\alpha \psi_\alpha
=\sum_{\mathcal{P}_G e_\beta=e_\beta} \zeta_\beta e_\beta
+\sum_{\mathcal{P}_G e_\gamma=0} \eta_\gamma e_\gamma.
\end{equation}
%
Expansion of function around $q_o$ by eigenfunctions $e_\alpha$
defines action in terms of $\zeta, \eta$:
\begin{equation}
\label{actionInFormI}
S(\zeta,\eta)
=S\left[
	q_o
	+\sum_{\mathcal{P}_{G}e_\beta=e_\beta} \zeta_\beta e_\beta
	+\sum_{\mathcal{P}_{G}e_\gamma=0} \eta_\gamma e_\gamma
	\right].
\end{equation}
The invariance
$S[g q]=S[q]$ for arbitrary $q$
and $g q_o=q_o$
yields
another expression for the action:
\begin{equation}
\label{actionInFormII}
S(\zeta,\eta)
=S{\left[
	q_o
	+\sum_{\mathcal{P}_Ge_\beta=e_\beta} \zeta_\beta e_\beta
	+\sum_{\mathcal{P}_Ge_\gamma=0} \eta_\gamma\, g e_\gamma
	\right]}.
\end{equation}
The following well known theorem holds:
\begin{theorem}
\label{theoremProjectionOperator}
If a group $G$ is a symmetry group for $q_o$ and the action,
a stationary point in subspace $\mathcal{P}'_G=1$
is a stationary point in whole space,
namely,  a solution of the equations of motion.
\end{theorem}
\begin{proof}
Let $f=\sum_{\mathcal{P}_G\, e_\beta=e_\beta} \zeta_\beta e_\beta$,
then expansion of the action in \eref{actionInFormII} yields
\begin{eqnarray}
\fl
S(\zeta,\eta)
&=S[q_o]+ \sum \frac{\lambda_\beta}{2}\zeta_\beta^2
	+\sum \frac{\lambda_\gamma}{2}\eta_\gamma^2
	+\sum_{n\ge 3} \frac{1}{n!}\Braket{
		\Big(f+\sum_{\mathcal{P}_G e_\gamma=0}
		\eta_\gamma g e_\gamma\Big)^n
		}.\label{expansionII}
\end{eqnarray}
Since the eigenvalue and the inner product of eigenfunctions
are invariant of $g$, the second order terms 
in \eref{expansionII} are unchanged.
See \ref{secInvarianceOfIntegrals} for detail.
Then the first derivative of action  for $\eta_\gamma$
at all $\eta=0$ yields
\begin{equation}
\label{actionInG}
\left.\frac{\partial }{\partial \eta_\gamma}S[\zeta,\eta]\right|_{\eta=0}
=\sum_{n\ge 3} \frac{1}{(n-1)!}\Braket{
		f^{n-1}
		\, (1-\mathcal{P}_G)g e_\gamma
		}.
\end{equation}
Here we use the expression 
$e_\gamma=(1-\mathcal{P}_G)e_\gamma$
for $\mathcal{P}_G e_\gamma=0$.
Then the average for $g\in G$ yields
\begin{eqnarray}
\left.\frac{\partial }{\partial \eta_\gamma}S(\zeta,\eta)\right|_{\eta=0}
=\sum_{n\ge 3} \frac{1}{(n-1)!}\Braket{
		f^{n-1}
		\, (1-\mathcal{P}_G)\mathcal{P}_G e_\gamma
		}
=0.
\end{eqnarray}
Therefore, a stationary point in $\eta=0$ subspace
(namely $\mathcal{P}'_G=1$ subspace) is a stationary point
in whole space.
\end{proof}

Now, the following corollary is quite useful
\cite{sattinger, golubitsky, golubitskyII}.
\begin{corollary}
\label{corollaryProjectionOperator}
If a symmetry group $G$ for $q_o$ and $S$ exists 
such that
the solution $\mathcal{P}_G\phi(\theta)=\phi(\theta)$ is one dimension,
a bifurcation occurs in this dimension and 
the bifurcated solution has symmetry of $\mathcal{P}_G q_b=q_b$.
\end{corollary}
The statement
``the solution $\mathcal{P}_G\phi(\theta)=\phi(\theta)$ is one dimension''
means that
there are only two solutions $\theta=\theta_0$ and $\theta_0+\pi$
for this equation.
\begin{proof}
Since $\mathcal{P}_G\phi(\theta)=\phi(\theta)$ is one dimension,
the eigenvalue $\kappa$ is not degenerate
in  $\mathcal{P}'_G=1$ subspace.
Then, Lyapunov-Schmidt reduction yields 
one variable function $S_{LS}(r)$
for this subspace,
\begin{equation}
	S_{LS}(r)=S\left[q_o+ r\mathcal{P}_G\phi(\theta)
		+r\sum_{\mathcal{P}_G\psi_\alpha=\psi_\alpha}
			\epsilon_\alpha \psi_\alpha
			\right].
\end{equation}
As shown in \sref{secNonDegeneratedCase},
a non-trivial stationary point $q_b$ exists in this subspace.
Then, by the theorem \ref{theoremProjectionOperator},
$q_b$  is a bifurcated solution.
Obviously, $\mathcal{P}_G q_b=q_b$ follows.
\end{proof}

For degenerate case,
the projection operator $\mathcal{P}_G$ picks up
one direction $\theta_0$ by
$\mathcal{P}_G\phi(\theta_0)=\phi(\theta_0)$.
This is the answer to the question in subsection
\ref{secDegeneratedCase}.
Examples will be shown in the following subsections.


\subsection{Inheritance of symmetry}
\label{secInheritance}
The invariance of $q_o$ and action $S$
under the symmetry group $G$
is inherited by integrals  $\braket{\dots}$ and reduced action $S_{LS}$.
Here we list two useful equalities.

\begin{lemma}
\label{lemmaInvariance}
For any element $g$ of a symmetry group $G$ for $q_o$ and the action $S$,
the following equality holds,
\begin{equation}
\label{invarianceOfIntegrals}
\braket{e_i\vphantom{)}^{n_i}e_j\vphantom{)}^{n_j}\dots}
=\braket{(ge_i)^{n_i}(g e_j)^{n_j}\dots},
\end{equation}
where $e_i$ are any eigenfunctions of $\mathcal{H}$,
and $n_i$ are any natural numbers.
\end{lemma}

\begin{theorem}
\label{theoryInvarianceOfSLS}
For any element $g$ of a symmetry group $G$ for $q_o$ and the action $S$,
the reduced action $S_{LS}(r,\theta)$ inherits the invariance of the action.
Namely
\begin{equation}
\label{invarianceSLS}
\mbox{if } g\, r\phi(\theta)=r'\phi(\theta'),
\mbox{ then }
S_{LS}(r,\theta)=S_{LS}(r',\theta').
\end{equation}
\end{theorem}

In this theorem, $r'$ is one of $\pm r$.
For one dimensional case,
the theorem means that ``if $g\phi=-\phi$ then $S_{LS}(r)=S_{LS}(-r)$''.
For two dimensional case, this theorem may mean that
``if $g\phi(\theta)=\phi(\theta')$ then $S_{LS}(r,\theta)=S_{LS}(r,\theta')$''.
However,
when $g\phi(\theta)=-\phi(\theta)$,
it is convenient to allow $r<0$ and 
read this theorem
``if $g\phi(\theta)=-\phi(\theta)$ then $S_{LS}(r,\theta)=S_{LS}(-r,\theta)$''.
Since $(r,\theta)$ and $(-r, \theta+\pi)$ represent
the same function $r(\cos(\theta)\phi_1+\sin(\theta)\phi_2)$, the following identity is satisfied:
\begin{equation}
S_{LS}(r,\theta)=S_{LS}(-r,\theta+\pi).
\end{equation}


Proofs are given in \ref{secInvarianceOfIntegrals},
because the meaning of these equalities are clear while  proofs are long.


\section{Bifurcations  of $D_6$}
\label{bifurcationOfD6}
Since 
a
figure-eight solution is invariant under the $D_6$ transformations
that is defined by \eref{eqDefD6},
the Hessian $\mathcal{H}$ is also invariant under $D_6$.
Then the eigenvalues and eigenfunctions are classified by 
irreducible representations of $D_6$.
It is well known that group $D_6$ has six irreducible representations,
4 one-dimensional representations and 2 two-dimensional ones.
Each representation is specified by the eigenvalues
$\mathcal{P}_C',\, \mathcal{M}'$ and $\mathcal{S}'$ of operators
$\mathcal{P}_C,\, \mathcal{M}$ and $\mathcal{S}$.
Since $\mathcal{P_C}^2=\mathcal{P_C}$ and $\mathcal{M}^2=\mathcal{S}^2=1$,
the eigenvalues are $\mathcal{P_C}'=1$ or $0$, $\mathcal{M}'=\pm 1$
and $\mathcal{S}'=\pm 1$.
The original solution $q_o$ has 
$\mathcal{P_C}'=\mathcal{M}'=\mathcal{S}'=1$ by definition.
Table \ref{table6RepsOfD6} shows the six irreducible representations.
In the following sections \ref{subTypeI} to \ref{subTypeVI},
bifurcation patterns for each irreducible representation
will be described.
The results 
are summarized in the table \ref{table6RepsOfD6} and 
the figure \ref{figbifurcationsAndSymmetryBreaking}.

In this section,
we treat bifurcations of $D_6$.
Therefore, a symmetry group $G$ is always a subgroup of $D_6$.
The condition $g q_o=q_o$ and $S[g q]=S[q]$ is always satisfied by $g \in G$.
Moreover, for an eigenfunction $e$,
a symmetry for $\{e,q_o, S\}$ represents
a group element $g$ of $G$ with $g e=e$, $g q_o=q_o$ and $S[g q]=S[q]$.

\begin{table}
\caption{\label{table6RepsOfD6}
Six irreducible representations of group $D_6$
characterized by $\mathcal{P_C}'$, $\mathcal{M}'$, and $\mathcal{S}'$.
The column $d$ represents dimension for the representation.
The last 
four
columns will be shown in each subsections
\ref{subTypeI} to \ref{subTypeVI}.
Columns $\mathcal{P}$ and $G$ represent 
the projection operator $\mathcal{P}$ and symmetry group $G$ for the bifurcated solution, respectively.
The ``Order'' represents order $n$ of the bifurcation: 
$\kappa=-A_{n+2}\, r^n/(n+1)!+O(r^{n+1}),\, A_{n+2}\ne0$.
Column ``Type'' represents type of bifurcation.}
\footnotesize
\begin{tabular}{c|rrr|c|llcl}
\br
Representation & $\mathcal{P_C}'$ & $\mathcal{M}'$ & $\mathcal{S}'$ & $d$ 
& $\mathcal{P}$ &$G$&Order&Type\\
\hline
I & 1 & 1 & 1 & 1 & $\mathcal{P_C}\mathcal{P_M}\mathcal{P_S}$&$D_6$
&1&fold$^{\rm b}$\\
II & 1 & 1 & $-1$ & 1 & $\mathcal{P_C}\mathcal{P_M}$&$C_6$
&2&one-side \\
III & 1 & $-1$ & 1 & 1 & $\mathcal{P_C}\mathcal{P_S}$&$D_3$
&2&one-side \\
IV & 1 & $-1$ & $-1$ & 1 & $\mathcal{P_C}\mathcal{P_{MS}}$ &$D_3'$
&2&one-side\\
\hline
V & 0 & 1 & $\pm 1$ & 2& $\mathcal{P_M}\mathcal{P_S}$
&$D_2$ &1&both-sides\\
VI & 0 & $-1$ & $\pm 1$ & 2 
	& $\mathcal{P_S}$ or$^{\rm a}$ $\mathcal{P_{MS}}$
	&$D_1$ or$^{\rm a}$ $D_1'$
	&2&double$^{\rm a}$ one-side\\
\br
\end{tabular}

$^{\rm a}$In the representation VI,
bifurcation yields two different kind of bifurcated solutions:
$\mathcal{P_S}$ invariant solution
and $\mathcal{P_{MS}}$ invariant one.\\
$^{\rm b}$Fold bifurcation is suitable, although both-sides is still possible. 
\end{table}

\begin{figure}
   \centering
\includegraphics[width=13cm]{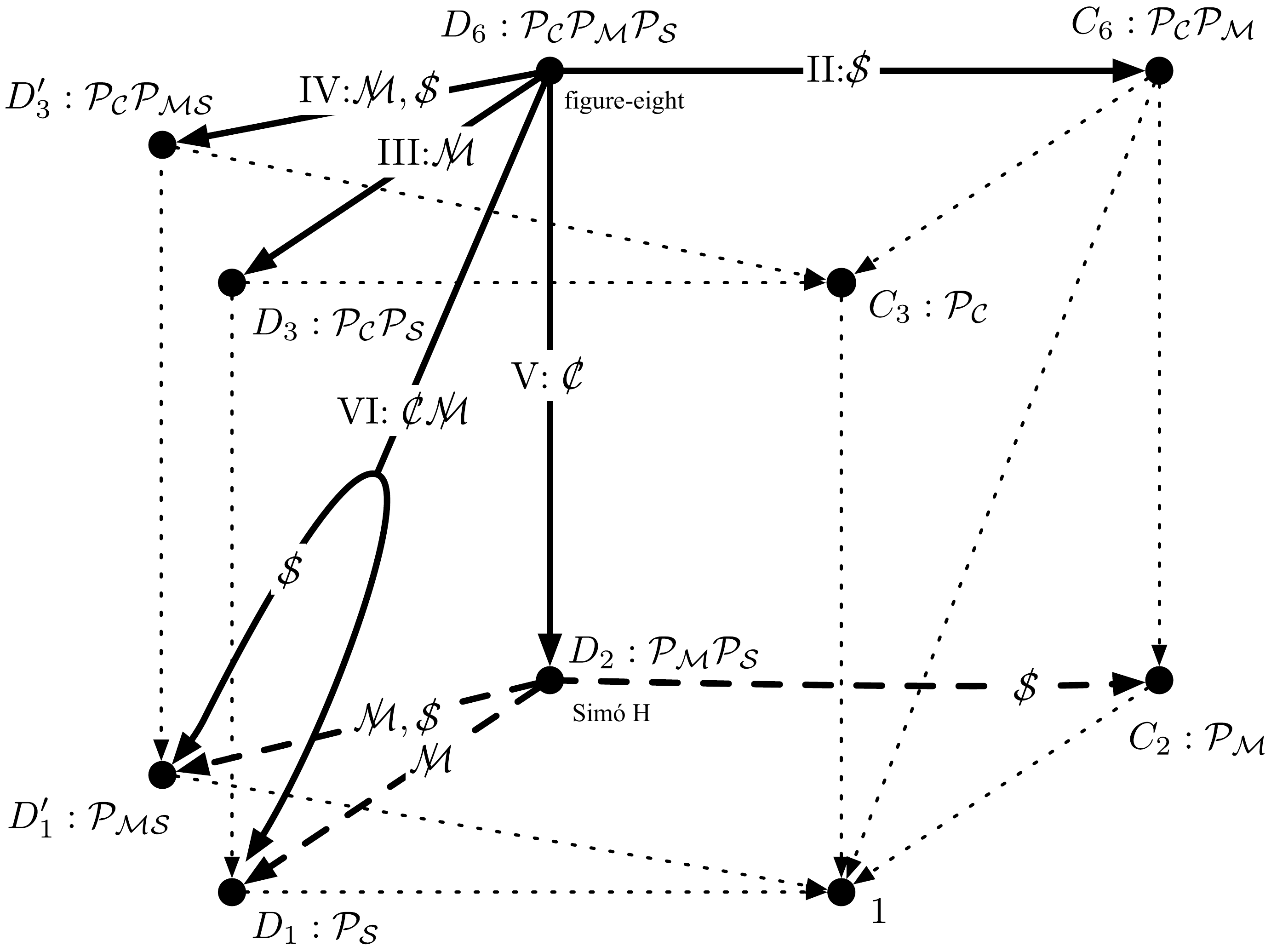}
   \caption{Bifurcations and symmetry breaking of
   $D_6$ and $D_2$.
   Bifurcations of $D_6$ and $D_2$ are represented by
   thick arrows and dashed thick arrows respectively.
   Each vertex represents the solution of the equations of motion.
   The symbol ``$G:\mathcal{P}$'' at each vertex represents
   symmetry group $G$ and projection operator $\mathcal{P}$ for the solution.
   The symbol ``No:$\slashed{\mathcal{O}}$'' on the arrows represents
   the  number of irreducible representation in  table \ref{table6RepsOfD6}
   and broken symmetry $\mathcal{O}$.
   The fork in VI shows that
   this bifurcation yields two bifurcated solutions,
   one with symmetry group $D_1$ and the other with $D'_1$.
   The bifurcation I in  table \ref{table6RepsOfD6}
    is bifurcation of $D_6 \to D_6$
    which is omitted in this figure.
   }
   \label{figbifurcationsAndSymmetryBreaking}
\end{figure}

\subsection{Representation I: $\mathcal{P_C}'=\mathcal{M}'=\mathcal{S}'=1$}
\label{subTypeI}
Irreducible representation I is 
characterized by $\mathcal{P_C}'=\mathcal{M}'=\mathcal{S}'=1$,
namely,
eigenfunction $\phi$
has these eigenvalues that is the same as $q_o$.
The representation for group elements are $\tilde\mathcal{B}=\tilde\mathcal{S}=1$.
This representation is called identity representation or 
trivial representation.

Since all  $\mathcal{C},\, \mathcal{M}, \mathcal{S}$ are 
the symmetry for $\{\phi, q_o,S\}$,
the projection operator is
\begin{equation}
\mathcal{P}_{I}
=\mathcal{P_C}\mathcal{P_M}\mathcal{P_S}.
\end{equation}
Then by the corollary \ref{corollaryProjectionOperator}, 
a bifurcation occurs
and the bifurcated solution $q_b$ has the same invariance $D_6$
as the original $q_o$.
The bifurcation pattern by this representation is
\begin{equation}
D_6 \to D_6.
\end{equation}

Since there is no reason for $\braket{\phi^3}=0$,
we can safely assume $A_3=\braket{\phi^3}\ne 0$.
Then order 1 bifurcation described
in \eref{eqFirstOrderRb}, \eref{eqFirstOrderSLS}
and \eref{eqFirstOrderSLSSecond} occurs.

As stated in \sref{secNonDegeneratedCase},
order $1$ bifurcation can describe 
a fold bifurcation or ``both-sides'' bifurcation.
Since symmetry group for $q_o$ and $q_b$ is the same
in this bifurcation,
a fold bifurcation is suitable.

\subsection{Representation II: $\mathcal{P_C}'=\mathcal{M}'=1$ and $\mathcal{S}'=-1$}
\label{subTypeII}
Irreducible representation  II is 
characterized by $\mathcal{P_C}'=\mathcal{M}'=1$ and $\mathcal{S}'=-1$,
namely, $\phi$
has these eigenvalues.
In this representation $\tilde\mathcal{B}=1$, and $\tilde\mathcal{S}=-1$.
Since the symmetry for $\{\phi, q_o,S\}$ are $\mathcal{C}$ and $\mathcal{M}$,
the projection operator is
\begin{equation}
\mathcal{P}_{II}=\mathcal{P_C}\mathcal{P_M}.
\end{equation}
By the corollary \ref{corollaryProjectionOperator},
a bifurcation occurs and the invariance of the bifurcated solution is 
$\mathcal{P}_{II}q_b=q_b$.
Namely, the invariance of $q_b$ is
\begin{equation}
C_6=\{1,\mathcal{B},\mathcal{B}^2,\dots,\mathcal{B}^5\}.
\end{equation}
The symmetry for $\mathcal{S}$ is broken.
Bifurcation pattern by this representation is
\begin{equation}
D_6 \to C_6.
\end{equation}

By $\mathcal{S}\phi=-\phi$
and the theorem \ref{theoryInvarianceOfSLS},
the reduced action is even function of $r$:
$S_{LS}(-r)=S_{LS}(r)$.
Therefore, $A_{2 n+1}=0$ for $n \ge 1$.
There is no reason for $A_4=0$, we assume $A_4\ne 0$.
Therefore  order 2 bifurcation
described by \eref{order2Rb}, \eref{order2SLS} and \eref{order2S''LS}
occurs.

\subsection{Representation III: $\mathcal{P_C}'=1,\, \mathcal{M}'=-1,\, \mathcal{S}'=1$}
\label{subTypeIII}
Irreducible representation   III is 
characterized by $\mathcal{P_C}'=\mathcal{S}'=1$ and $\mathcal{M}'=-1$.
In this representation $\tilde\mathcal{B}=-1$ and $\tilde\mathcal{S}=1$.
Since the symmetry for $\{\phi, q_o,S\}$
are $\mathcal{C}$ and $\mathcal{S}$,
the projection operator  is
\begin{equation}
\mathcal{P}_{III}=\mathcal{P_C}\mathcal{P_S}.
\end{equation}
By the corollary \ref{corollaryProjectionOperator},
a bifurcation occurs and
the invariance of $q_b$ is
\begin{equation}
\label{D3}
D_3=\{1,\mathcal{C},\mathcal{C}^2,\mathcal{S},\mathcal{S}\mathcal{C},\mathcal{S}\mathcal{C}^2\}.
\end{equation}
The symmetry for $\mathcal{M}$ is broken.
The bifurcation pattern is
\begin{equation}
D_6 \to D_3.
\end{equation}

Since $\mathcal{M}\phi=-\phi$, $S_{LS}(-r)=S_{LS}(r)$
and $A_{2 n+1}=0$ for $n \ge 1$.
Assuming $A_4\ne 0$,
order 2 bifurcation
described by \eref{order2Rb}, \eref{order2SLS} and \eref{order2S''LS}
occurs.

\subsection{Representation IV: $\mathcal{P_C}'=1,\, \mathcal{M}'=\mathcal{S}'=-1$}
\label{subTypeIV}
Irreducible representation IV is 
characterized by $\mathcal{P_C}'=1,\, \mathcal{M}'=\mathcal{S}'=-1$.
In this representation $\tilde\mathcal{B}=\tilde\mathcal{S}=-1$.
Since the symmetry for $\{\phi, q_o,S\}$
are $\mathcal{C}$ and $\mathcal{MS}$,
the projection operator is
\begin{equation}
\mathcal{P}_{IV}=\mathcal{P_C}\mathcal{P_{MS}}.
\end{equation}
By the corollary \ref{corollaryProjectionOperator},
a bifurcation occurs and  invariance of $q_b$ is
\begin{equation}
D_3'=\{1,\mathcal{C},\mathcal{C}^2,
	\mathcal{MS},\mathcal{MS}\mathcal{C},\mathcal{MS}\mathcal{C}^2\}.
\end{equation}
Here, $'$ is added to distinguish this dihedral group of 6 elements
from $D_3$ in \eref{D3}.
The symmetry of both $\mathcal{M}$ and $\mathcal{S}$ are broken,
while the symmetry of $\mathcal{MS}$ is unbroken.

Since $\mathcal{M}\phi=\mathcal{S}\phi=-\phi$, 
$S_{LS}(-r)=S_{LS}(r)$ and $A_{2 n+1}=0$ for $n \ge 1$.
Assuming $A_4\ne 0$,
order 2 bifurcation
described by \eref{order2Rb}, \eref{order2SLS} and \eref{order2S''LS}
occurs.

\subsection{Representation V: $\mathcal{P_C}'=0,\, \mathcal{M}'=1$ and $\mathcal{S}'=\pm1$}
\label{subTypeV}
Irreducible representation V is 
characterized by $\mathcal{P_C}'=0,\, \mathcal{M}'=1$ and $\mathcal{S}'=\pm1$.
This is two  dimensional representation with
\begin{equation}
\tilde{\mathcal{B}}=\left(
	\begin{array}{rr}
		\cos(2\pi/3) & -\sin(2\pi/3) \\
		\sin(2\pi/3) & \cos(2\pi/3)
\end{array}
\right)
\mbox{ and }
\tilde{\mathcal{S}}=\left(
	\begin{array}{cc}
		1 & 0 \\
		0 & -1
\end{array}
\right).
\end{equation}
Eigenvalue $\kappa$ has degeneracy $d=2$.
Let the eigenfunction $\phi_\pm$ be the eigenfunction
of $\mathcal{S}\phi_\pm=\pm \phi_\pm$.
They have $\mathcal{P_C}'=0$ and  $\mathcal{M}'=1$.

Since the symmetry for $\{\phi_+,q_o,S\}$ are
$\mathcal{M}$ and $\mathcal{S}$,
the projection operator for $\phi_+$ is
\begin{equation}
\label{PV}
\mathcal{P}_{V+}=\mathcal{P_M}\mathcal{P_S}.
\end{equation}
Note that this projection operator chooses $\phi_+$
and discards  $\phi_-$:
$\mathcal{P}_{V+}\phi_+=\phi_+$, $\mathcal{P}_{V+}\phi_-=0$.
Therefore
by the corollary \ref{corollaryProjectionOperator},
a bifurcation in subspace of $\mathcal{P}_{V+}=\mathcal{P_M}\mathcal{P_S}$ occurs.
The broken symmetry is $\mathcal{C}$ and 
invariance of the bifurcated solution is
\begin{equation}
D_2=\{1,\mathcal{M},\mathcal{S},\mathcal{SM}\}.
\end{equation}
The bifurcation pattern is
\begin{equation}
D_6 \to D_2.
\end{equation}

On the other hand,
the symmetry for $\{\phi_-, q_o,S\}$ is
only $\mathcal{M}$.
Therefore, the projection operator is $\mathcal{P_M}$.
Since $\mathcal{P_M}$ does not exclude $\phi_+$,
subspace of $\mathcal{P_M}$ remains two-dimensional. 
So the corollary \ref{corollaryProjectionOperator} doesn't ensure
a bifurcation in  the direction of $\phi_-$.
Indeed we can show that there is bifurcated solution in $\phi_+$ direction,
whereas no bifurcated solution in $\phi_-$ direction,
by explicitly calculating the reduced action $S_{LS}$.

Let $\phi(\theta)=\cos(\theta)\phi_++\sin(\theta)\phi_-$,
then as shown in \ref{AkForV},
the the reduced action is given by
\begin{eqnarray}
S_{LS}(r,\theta)
	=\frac{\kappa}{2}r^2
		+\frac{A_3(0)}{3!}r^3\cos(3\theta)
		+\frac{A_4(0)}{4!}r^4
		+O(r^5),\label{SLSforV}\\
\frac{A_3(0)}{3!}=\frac{1}{3!}\braket{\phi_+^3},\label{A3V}\\
\frac{A_4(0)}{4!}=
	\frac{1}{4!}\braket{\phi_+^4}
		-\frac{1}{(2!)^3}\sum_\alpha \frac{\braket{\phi_+^2\psi_{\alpha+}}^2}{\lambda_\alpha}
		\label{A4V},
\end{eqnarray}
where $\mathcal{S}\psi_{\alpha+}=\psi_{\alpha+}$.
The reduced action $S_{LS}$ in $(x,y)$ plane is shown in \fref{figSLSforV},
where $(x,y)$ are orthogonal coordinates 
defined by $x=r\cos(\theta)$ and $y=r\sin(\theta)$ as usual.
Then $r\phi(\theta)=x\phi_++y\phi_-$.
\begin{figure}
   \centering
\includegraphics[width=4cm]{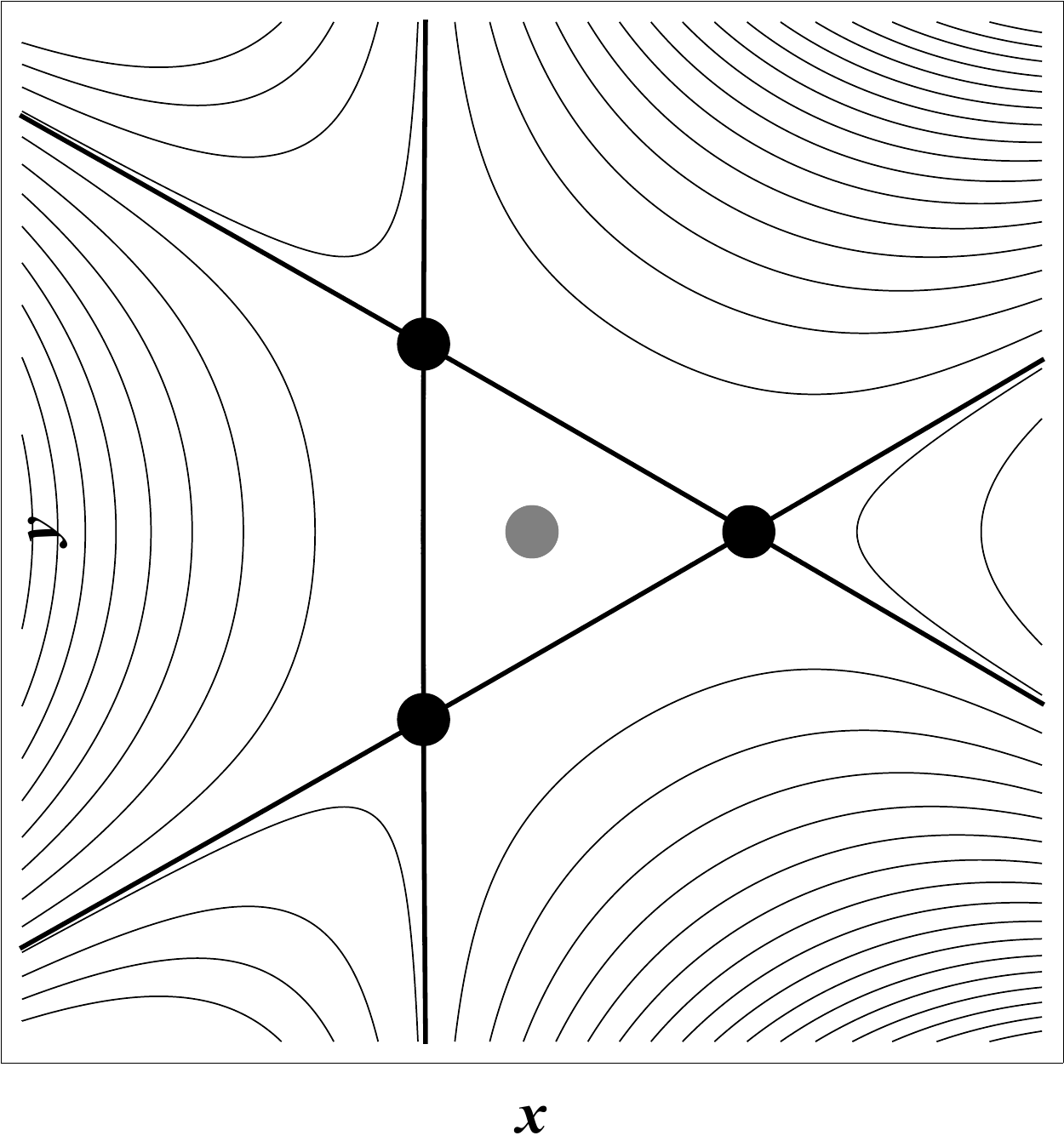}
\includegraphics[width=4cm]{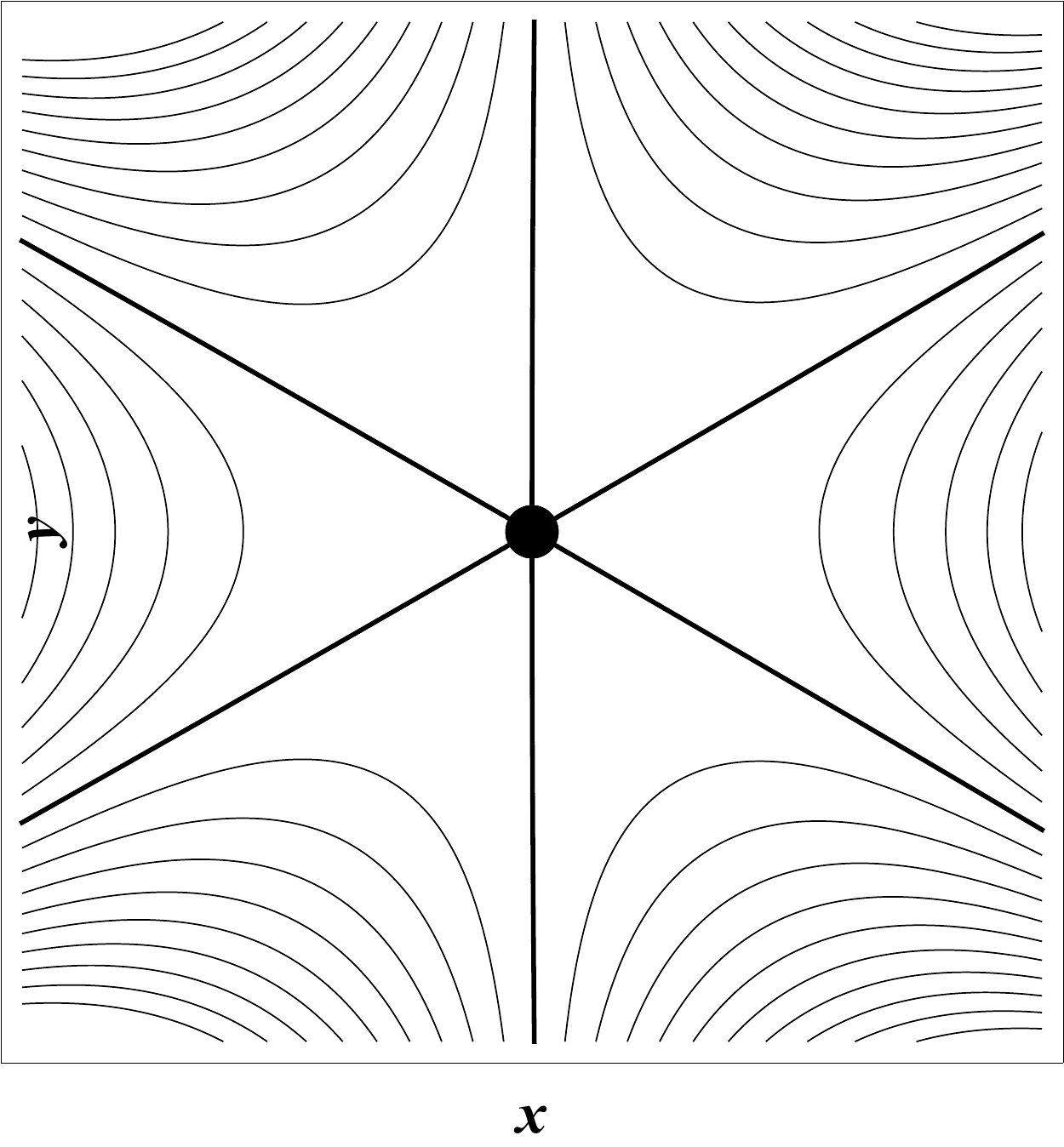}
\includegraphics[width=4cm]{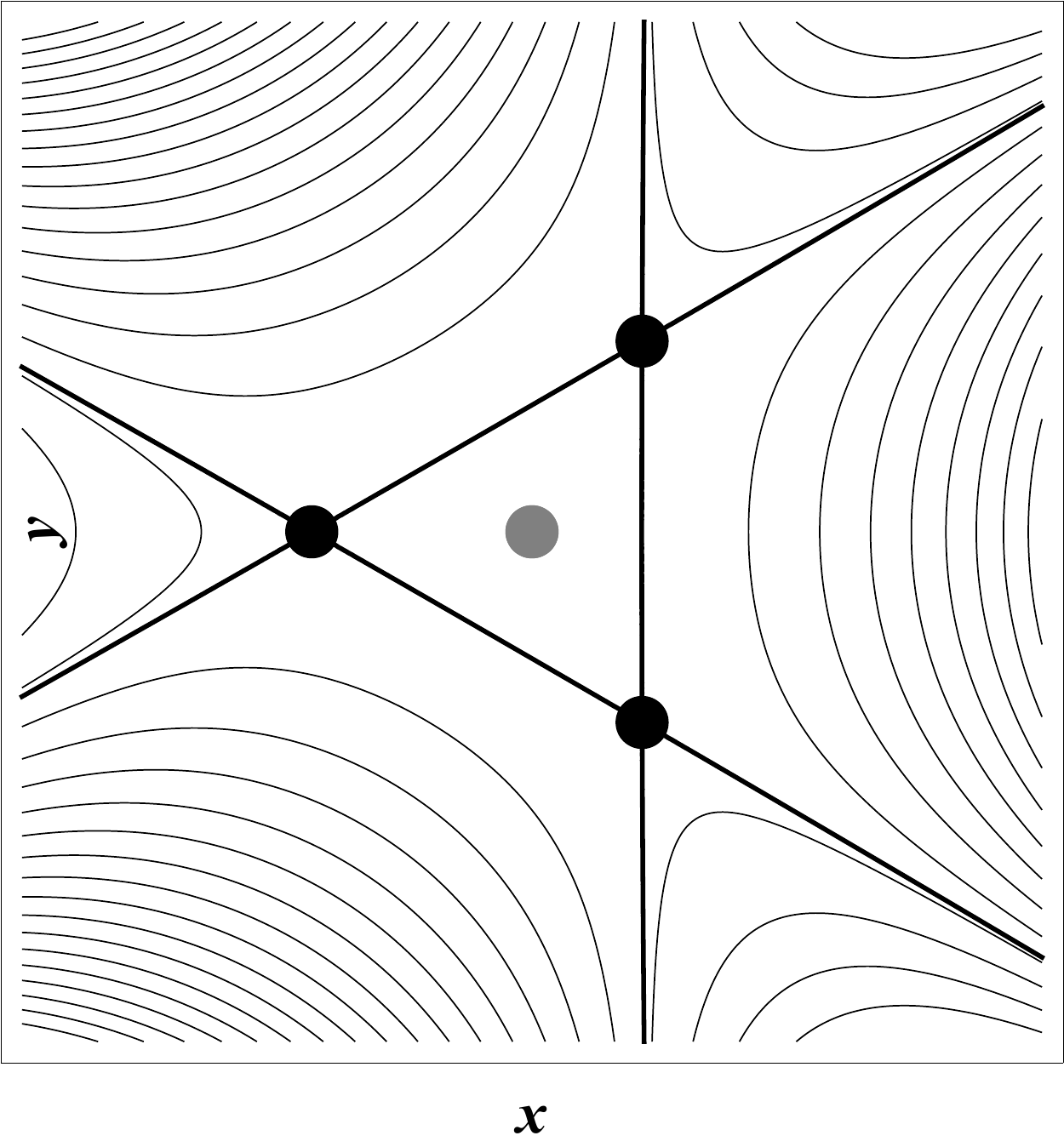}\\[1mm]
\includegraphics[width=4cm]{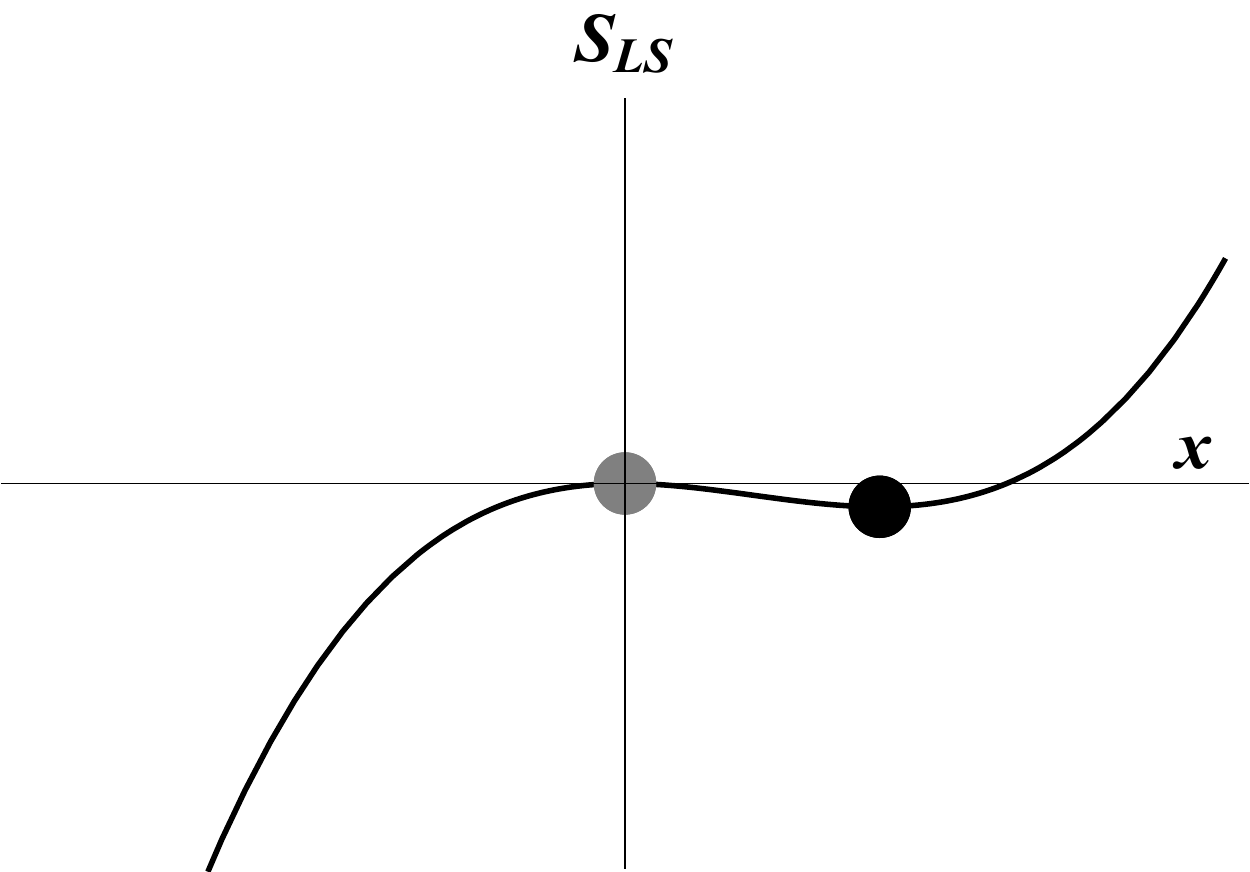}
\includegraphics[width=4cm]{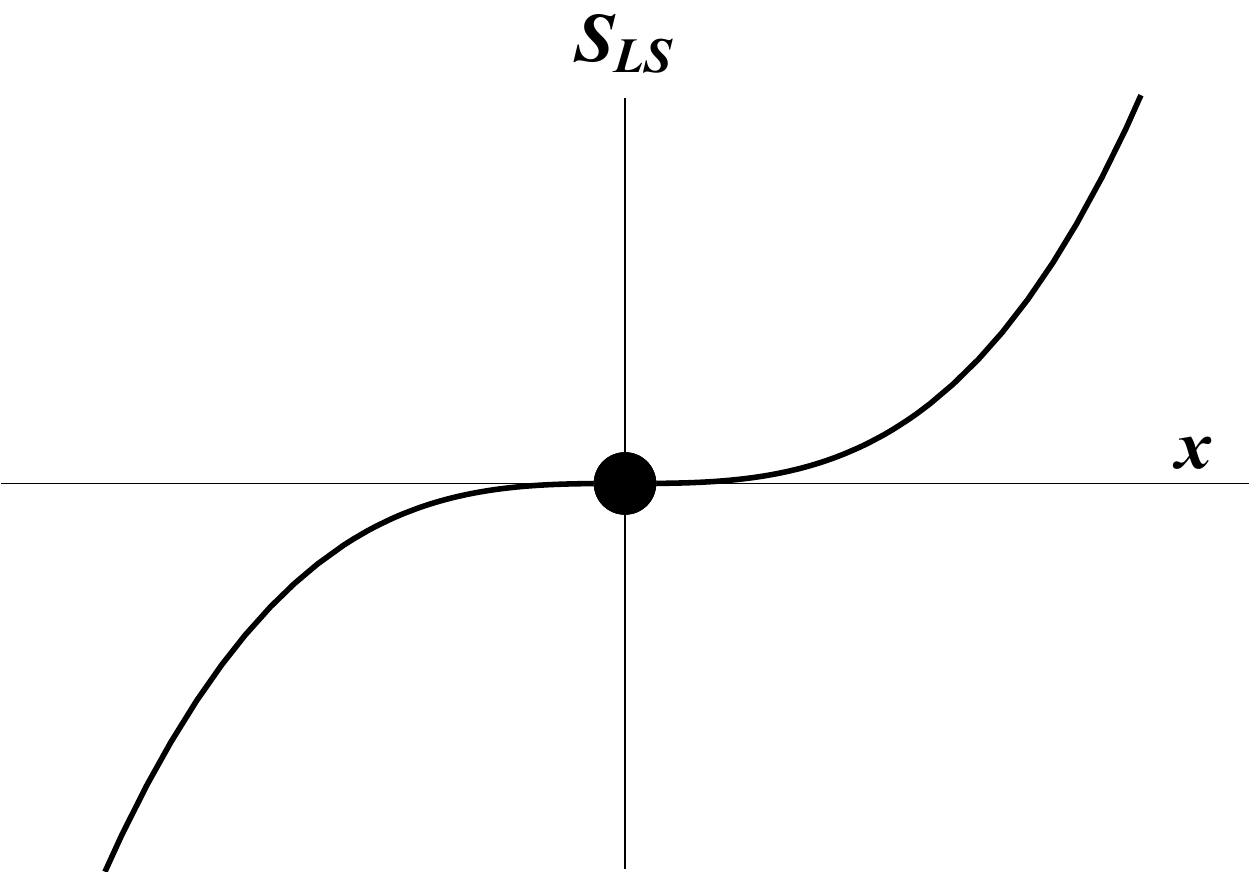}
\includegraphics[width=4cm]{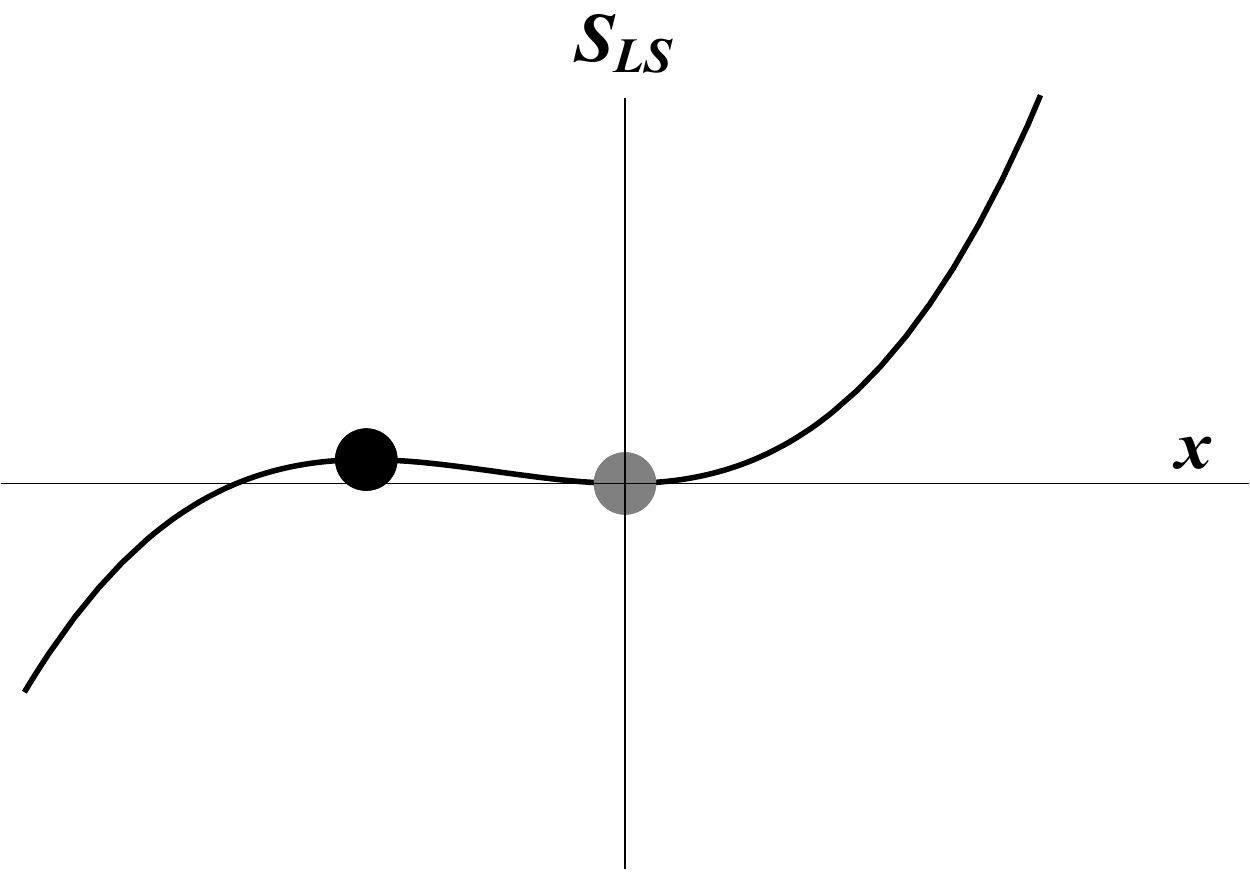}
   \caption{Reduced action $S_{LS}$
   of representation V for $A_3, A_4>0$.
   Upper: Contour plot of $S_{LS}$ 
   in orthogonal coordinates $(x,y)$.
   Lower: $S_{LS}$ for $y=0$.
   From left to right $\kappa<0, \kappa=0$ and $\kappa>0$.
   Gray circle and black circles represent $q_o$
   and $q_b, \mathcal{C}q_b, \mathcal{C}^2 q_b$
   respectively.
   For sufficiently small  $\kappa$ and short range of $r$,
   terms of $O(r^4)$ have no effect.
   }
   \label{figSLSforV}
\end{figure}

Note that $S_{LS}$ has the symmetry of regular triangle $D_3$,
instead of $D_6$.
The reason is $\mathcal{M}\phi(\theta)=\phi(\theta)$.
By this invariance, the theorem \ref{theoryInvarianceOfSLS}
gives an identity: $S_{LS}(r,\theta)=S_{LS}(r,\theta)$.
Therefore, $\mathcal{M}$ invariance is invisible in $S_{LS}$.
On the other hand,
$\mathcal{S}\phi(\theta)=\phi(-\theta)$,
$\mathcal{C}\phi(\theta)=\phi(\theta+2\pi/3)$
and  the theorem \ref{theoryInvarianceOfSLS}
give
$S_{LS}(r,-\theta)=S_{LS}(r,\theta)$ and
$S_{LS}(r,\theta\pm 2\pi/3)=S_{LS}(r,\theta)$
that show apparent invariance of $S_{LS}$ in $D_3$.

In other words, 
since quotient groups of $D_6$ and $D_2$ by
the centre $\{1,\mathcal{M}\}$ are
\begin{eqnarray}
D_6/\{1,\mathcal{M}\}
\cong \{1,\mathcal{C}, \mathcal{C}^2,\mathcal{S},
		\mathcal{S}\mathcal{C}, \mathcal{S}\mathcal{C}^2\}
=D_3\\
D_2/\{1,\mathcal{M}\}
\cong \{1,\mathcal{S}\}
=D_1,
\end{eqnarray}
this bifurcation pattern $D_6 \to D_2$ that keeps $\mathcal{M}$ symmetry
is equivalent to the bifurcation
\begin{equation}
D_3 \to D_1.
\end{equation}
This is the reason why
the reduced action $S_{LS}$ has $D_3$ symmetry.
The $D_3$ symmetry determines the form of $S_{LS}$
in \eref{SLSforV}.

The equations for stationary points are
\begin{eqnarray}
\partial_r S_{LS}(r,\theta)
	=r\left(
		\kappa+\frac{A_3(0)}{2}r\cos(3\theta)
		+\frac{A_4(0)}{6}r^2+O(r^3)
	\right)=0,\\
\partial_\theta S_{LS}
	=-\frac{A_3(0)}{2}r^3 \sin(3\theta)+O(r^5)=0.\label{dThetaSV}
\end{eqnarray}
The solutions are $r_o=0$ and
\begin{equation}
\label{rbV}
r_b=-\frac{2}{A_3(0)}\kappa-\frac{4A_4(0)}{3 A_3(0)^3}\kappa^2
		+O(\kappa^3),\,
\theta_k=\frac{2\pi}{3}k,\, k=0,1,2.
\end{equation}
This is order 1 bifurcation.
Let $q_b$ be the solution at $\theta=0$.
Then
\begin{equation}
	q_b = q_o+r_b \phi_+ + r_b \sum_\alpha \epsilon_\alpha \psi_\alpha,
\end{equation}
with $\mathcal{P}_{V+}q_b=q_b$.
Note that the solutions for  $\theta$ in \eref{rbV} are exact,
namely, $O(r^5)$ term in \eref{dThetaSV} does not change the solution $\theta$.
Because of the $D_3$ symmetry of the reduced action,
$\theta=0$ is exactly fixed. 
This is expected by the condition $\mathcal{P}_{V+}q_b=q_b$ in \eref{PV}.

The solutions corresponds to $k=1,2$ are just 
copies of $q_b$: $\mathcal{C}^k q_b$.
This is a direct result of $D_3$ symmetry of $S_{LS}(r,\theta)$.
The symmetry by $\mathcal{S}$ does not make new solution,
because $\mathcal{S}q_b=q_b$.

The action at the bifurcated solutions is
\begin{eqnarray}
S_{LS}(r_b,0)=S[q_b]-S[q_o]
=\frac{2}{3 A_3(0)^2}\kappa^3
	+\frac{2A_4}{3 A_3(0)^4}\kappa^4
	+O(\kappa^5)
\end{eqnarray}
and the second derivatives at the bifurcated solutions are 
$\left. r^{-1}\partial_r \partial_\theta S_{LS}\right|_{r=r_b,\theta=0}=0$ and
\begin{eqnarray}
\fl
\left.\partial_r^2 S_{LS}(r,\theta)\right|_{r=r_b,\theta=0}
	=\kappa+A_3(0)r_b+\frac{A_4(0)}{2}r_b^2+O(r_b^3)
	=-\kappa+\frac{2A_4(0)}{3 A_3(0)^2}\kappa^2
		+O(\kappa^3),\\
\fl
\left.r^{-2}\partial_\theta^2 S_{LS}(r,\theta)\right|_{r=r_b,\theta=0}
=-\frac{3}{2}A_3(0)r_b+O(r_b^3)
=3\kappa
	+\frac{2A_4(0)}{A_3(0)}\kappa^2+O(\kappa^3). \label{K2V}
\end{eqnarray}
Namely, the Hessian of the bifurcated solution has
non-degenerate eigenvalues $-\kappa+O(\kappa^2)$ and $3\kappa+O(\kappa^2)$.
Since $-\kappa$ and $3\kappa$ has opposite sign,
the bifurcated solutions are saddle.

As stated in \sref{secNonDegeneratedCase},
order $1$ bifurcation can describe 
fold bifurcation or ``both-sides'' bifurcation.
Since, symmetry group for $q_o$ and $q_b$ are different
and three copies $q_b$, $\mathcal{C}q_b$, $\mathcal{C}^2 q_b$
exist,
this  bifurcation should be ``both-sides''.

\subsection{Representation VI: $\mathcal{P_C}'=0,\, \mathcal{M}'=-1$ and $\mathcal{S}'=\pm1$}
\label{subTypeVI}
This is another two  dimensional representation with
\begin{equation}
\tilde{\mathcal{B}}=\left(
	\begin{array}{rr}
		\cos(\pi/3) & -\sin(\pi/3) \\
		\sin(\pi/3) & \cos(\pi/3)
\end{array}
\right)
\mbox{ and }
\tilde{\mathcal{S}}=\left(
	\begin{array}{cc}
		1 & 0 \\
		0 & -1
\end{array}
\right).
\end{equation}
This is the faithful representation of $D_6$.
Let $\phi_+$ and $\phi_-$ be the eigenfunctions with $\mathcal{S}\phi_\pm=\pm \phi_\pm$.
Then $\mathcal{P_S}\phi_+=\phi_+$ and $\mathcal{P_S}\phi_-=0$.
Moreover, $\mathcal{P_C}\phi_\pm=\mathcal{P_M}\phi_\pm=0$.
Note that
the function $\phi_-$ has invariance $\mathcal{P_{MS}}\phi_-=\phi_-$,
while $\mathcal{P_{MS}}\phi_+=0$.

The symmetry for $\{\phi_+, q_o,S\}$ is $\mathcal{S}$.
Therefore, the projection operator is
\begin{equation}
\mathcal{P}_{VI+}=\mathcal{P_S}.
\end{equation}
$\mathcal{P}_{VI+}$ excludes $\phi_-$: $\mathcal{P}_{VI+}\phi_-=0$.
Therefore
by the corollary \ref{corollaryProjectionOperator},
a bifurcation occurs for the direction of $\phi_+$
and invariance of $q_b$ is  $\mathcal{P}_{VI+}=\mathcal{P_S}$.
The broken symmetries are $\mathcal{C}$ and $\mathcal{M}$.
Since $\braket{\phi_+^3}=0$ by $\mathcal{M}\phi_+=-\phi_+$,
this bifurcation is order 2.
The bifurcation pattern in this direction is
\begin{equation}
\label{D1}
D_6 \to D_1=\{1, \mathcal{S}\}.
\end{equation}
We usually don't say $D_1$,
however, this notation is convenient for our purpose.
See \ref{BifurcationDn}.

On the other hand,
since the symmetry for $\{\phi_-, q_o,S\}$ is $\mathcal{MS}$,
the projection operator for them is
\begin{equation}
\mathcal{P}_{VI-}=\mathcal{P_{MS}}.
\end{equation}
Since $\mathcal{P}_{VI-}$ excludes $\phi_+$,
by the corollary \ref{corollaryProjectionOperator} a bifurcation occurs for the direction of $\phi_-$
and $q_b$ has invariance $\mathcal{P}_{VI-}=\mathcal{P_{MS}}$.
The symmetry for $\mathcal{C}$, $\mathcal{M}$, and $\mathcal{S}$
are all broken,
whereas the symmetry for $\mathcal{MS}$ remains.
The bifurcation pattern in this direction is
\begin{equation}
\label{D1prime}
D_6 \to D_1'=\{1, \mathcal{MS}\}.
\end{equation}
Here, $'$ is used again to distinguish it from $D_1$ in \eref{D1}.
Since $\braket{\phi_-^3}=0$ by $\mathcal{M}\phi_-=-\phi_-$,
this bifurcation is also order 2.

Therefore, at this bifurcation point, two order 2 bifurcations occur:
one for $\phi_+$ direction with $\mathcal{P_S}$ invariance
and another for $\phi_-$ direction with $\mathcal{P_{MS}}$ invariance.
Let us denote them by $q_{b\pm}$:
\begin{equation}
q_{b\pm}=q_o+r_{b\pm}\phi_\pm
	+r_{b\pm}\sum_{\mathcal{P}_{VI\pm}\psi_\alpha=\psi_\alpha}
			\epsilon_\alpha(r_{b\pm},\theta_\pm) \psi_\alpha.
\end{equation}
Then a question arises.
What is the relation between actions or second derivatives for one and another?
To see this relation, let us calculate $S_{LS}$ in this subspace.

Let $\phi(\theta)=\cos(\theta)\phi_++\sin(\theta)\phi_-$.
Then the reduced action should be
apparently invariant under the symmetry group $D_6$,
because any element of $D_6$ change $\phi(\theta)$.
Namely, $S_{LS}(r,\theta)$ is invariant under the transformations
$\theta \to \theta+2\pi k/6$, $k=0,1,2,\dots,5$ and $\theta \to -\theta$.
The $D_6$ invariance determines the form of $S_{LS}$ in the following:
\begin{equation}
\label{SforVI}
S_{LS}(r,\theta)
=\frac{\kappa}{2}r^2
	+\frac{A_4(0)}{4!}r^4
	+\frac{A_6(\theta)}{6!}r^6
	+O(r^8),
\end{equation}
\begin{equation}
A_6(\theta)=A_{6+}\cos(3\theta)^2+A_{6-}\sin(3\theta)^2,
\end{equation}
where $A_4(0)$, $A_{6_\pm}$ are independent from $\theta$.
The $\theta$ dependence of $S_{LS}(r,\theta)$ is
very small
because it appears in $r^6$ term.
Parts of direct derivation for \eref{SforVI} are shortly shown in \ref{AkForVI}.
\Fref{figSLSforVI} shows a typical behaviour of $S_{LS}$
in orthogonal coordinates $(x,y)=r(\cos\theta,\sin\theta)$.
\begin{figure}
   \centering
\includegraphics[width=6cm]{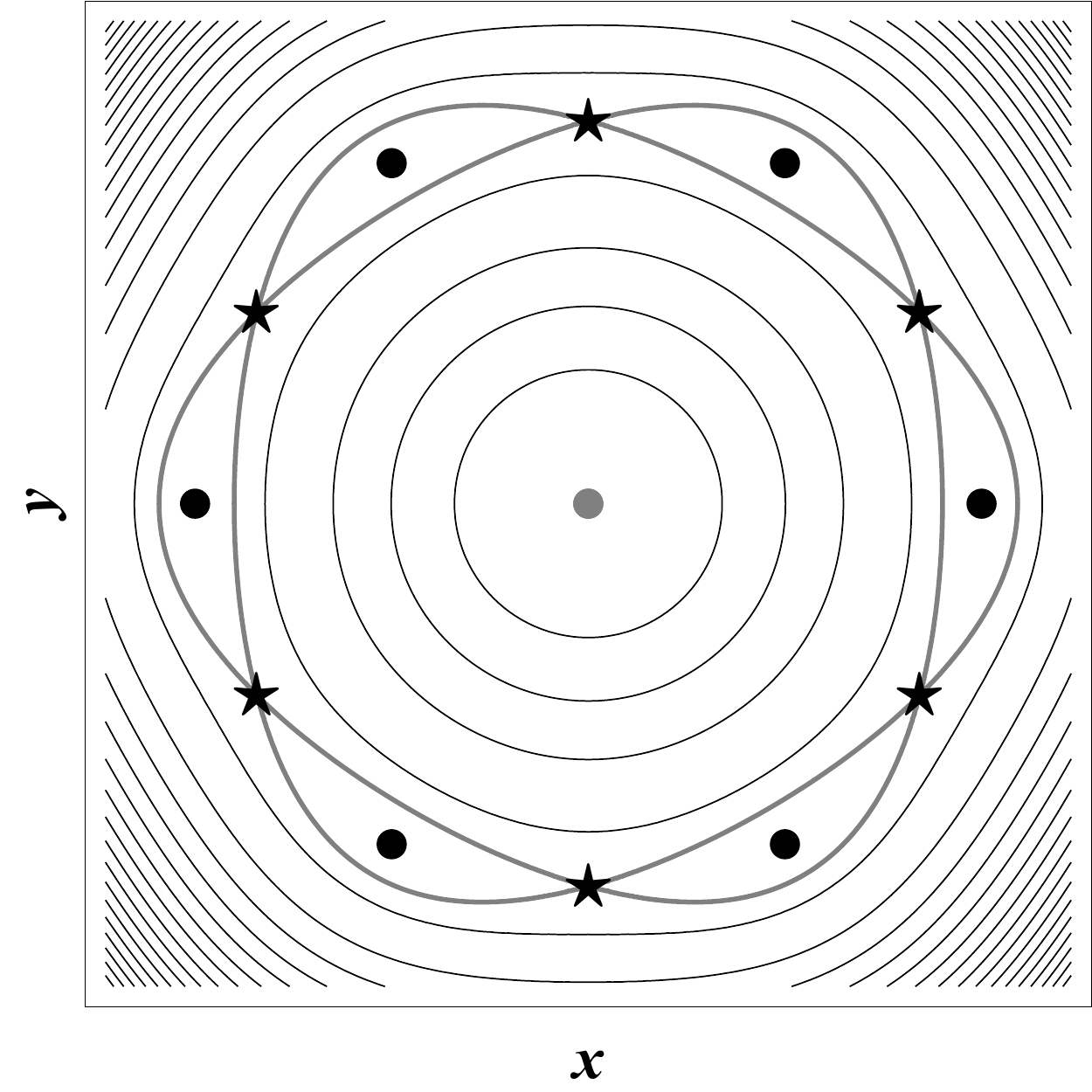}
\includegraphics[width=6cm]{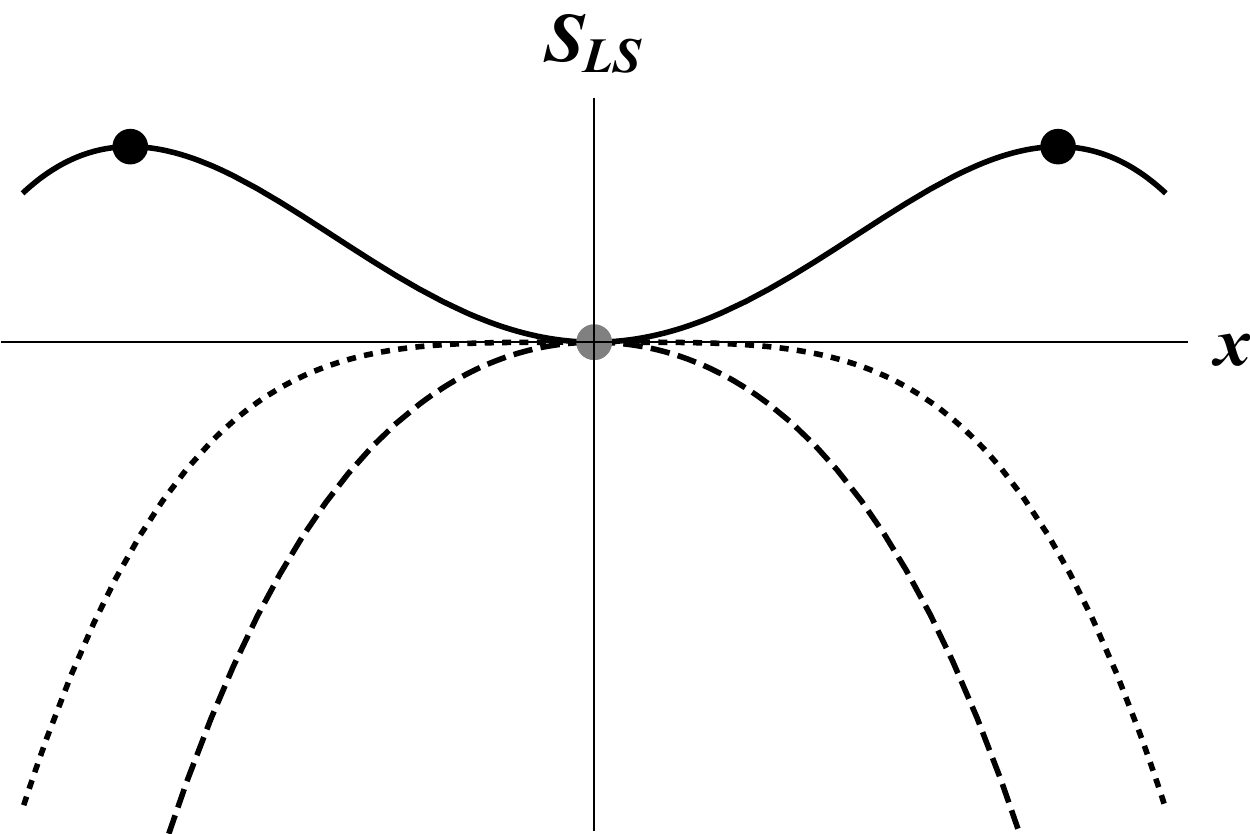}
   \caption{
   Reduced action $S_{LS}$ 
   of representation VI for 
   $A_4(0)<0, A_{6+}>A_{6-}>0$.
   Left: Contour plot of $S_{LS}$ for $\kappa>0$
   in orthogonal coordinates $(x,y)$.
   Gray circle at the centre is $q_0$
   and black circles and black stars are
   $\mathcal{B}^k q_{b+}$ and $\mathcal{B}^k q_{b-}$
   with $k=0,1,2,\dots,5$ respectively.
   For this assignment of $A_4(0)$ and 
   $A_{6\pm}$,
   $\mathcal{B}^k q_{b+}$ are local maximum 
   and $\mathcal{B}^k q_{b-}$ are saddle.
   Right:
   $S_{LS}$ for $y=0$.
   Three curves represent $\kappa$ negative (dashed), zero (dotted) 
   and positive(solid curve) respectively.
   For $A_4(0)<0$, bifurcated solutions $q_{b+}$ and $-q_{b+}$ (black circles)
   exist for $\kappa>0$.
   Gray circle at the origin is 
   $q_0$.
   }
   \label{figSLSforVI}
\end{figure}

The stationary points for $\theta$ are exactly given by
\begin{eqnarray}
\theta_{+k}=\frac{2\pi}{6}k,
\mbox{ and }
\theta_{-k}=\frac{\pi}{2}+\frac{2\pi}{6}k
\mbox{ with }k=0,1,2,\dots,5.
\end{eqnarray}
Let us write the solutions of $k=0$ as $\theta_{\pm}$ and $q_{b\pm}$,
then the other solutions are copies of them:
$\{q_{b+}, \mathcal{B}q_{b+},\dots,\mathcal{B}^5 q_{b+}\}$
and $\{q_{b-}, \mathcal{B}q_{b-},\dots,\mathcal{B}^5 q_{b-}\}$.
For $\theta_\pm$, the reduced actions are
\begin{equation}
\label{SLSforVI}
S_{LS}(r,\theta_\pm)
=\frac{\kappa}{2}r^2
	+\frac{A_4(0)}{4!}r^4
	+\frac{A_{6\pm}}{6!}r^6
	+O(r^8).
\end{equation}

Then the stationary points in $r$ are
$r=0$ and
\begin{eqnarray}
r_{b+}=\pm\left(\frac{-6\kappa}{A_4(0)}\right)^{1/2}
	\left(1+\frac{3A_{6+}}{20 A_4(0)^2}\kappa
		+O(\kappa^{2})
	\right),
\label{rbPlusForVI}\\
r_{b-}=\pm\left(\frac{-6\kappa}{A_4(0)}\right)^{1/2}
	\left(1+\frac{3A_{6-}}{20 A_4(0)^2}\kappa
		+O(\kappa^{2})
	\right).\label{rbMinusForVI}
\end{eqnarray}
Bifurcated solutions
appear one-side,
namely $\kappa>0$ if $A_4(0)<0$,
or $\kappa<0$ if $A_4(0)>0$.
They are order 2 bifurcations.
The values of action at the bifurcated solutions are
\begin{equation}
S_{LS}(r_{b\pm},\theta_{\pm})
=-\frac{3}{2 A_4(0)}\kappa^2
	-\frac{3A_{6\pm}}{10 A_4(0)^3}\kappa^3
	+O(\kappa^{4}).
\end{equation}
The sign of the first term is the same as $\kappa$
because $\kappa$ and $A_4(0)$ have opposite signs.
The difference of action between bifurcated solutions
are small,
because it appears in $\kappa^3$
term:
\begin{equation}
S_{LS}(r_{b+},\theta_{+})
-S_{LS}(r_{b-},\theta_{-})
=-\frac{3(A_{6+}-A_{6-})}{10 A_4(0)^3}\kappa^3+O(\kappa^{4}).
\end{equation}
Since $-\kappa^3/A_4(0)^3 >0$,
the sign of the difference  is the same  of that of 
$A_{6\pm}$.
The second derivatives at bifurcated solutions are
$r^{-1}\partial_r \partial_\theta S(r,\theta)|_{r_{b\pm}, \theta_\pm}
=0$ and
\begin{eqnarray}
\partial_r^2 S(r,\theta)|_{r_{b\pm}, \theta_\pm}
	=-2\kappa+\frac{3A_{6\pm}}{5 A_4(0)^2}\kappa^2+O(\kappa^{3}),\\
r^{-2}\partial_\theta^2 S(r,\theta)|_{r_{b\pm}, \theta_\pm}
	=\frac{9(A_{6\mp} -A_{6\pm})}{10 A_4(0)^2}\kappa^2+O(\kappa^3).
	\label{VIpatialThetaSquareS}
\end{eqnarray}
The $\pm$ symbol in the last equation should be read
$A_{6-}-A_{6+}$ for $\theta+$
and $A_{6+}-A_{6-}$ for $\theta-$.
Namely $A_6$ for other minus $A_6$ for here.
Since the main term of $\partial_r^2 S_{LS}$ for $q_{b\pm}$ are common,
while the main term of $r^{-2}\partial_\theta^2 S_{LS}$ has opposite sign,
one solution is a local minimum (if $A_4(0)>0$) or maximum (if $A_4(0)<0$),
while the other is saddle.

\section{Bifurcations of  figure-eight solutions}
\label{bifurcationsFigureEight}
A
figure-eight solution has $D_6$ symmetry and 
bifurcation patterns of $D_6$ are already described in section \ref{bifurcationOfD6}
based only on the algebraic structure of the group,
where the underlying symmetries for $\mathcal{B}$ and $\mathcal{S}$ have no meanings.
In this section, we describe the contents of the symmetry of bifurcated solutions
for figure-eight solutions.

Obviously,
the symmetry $\mathcal{C}=\sigma^{-1}R^{1/3}$ describes
choreographic symmetry.
Other symmetry described by $\mathcal{M}$, $\mathcal{S}$ and
$\mathcal{MS}$ is connected to geometric symmetries of locus of solutions.
Here, locus is defined by neglecting time and exchange of particles.
Then,
$\mathcal{M}=\mu_x R^{1/2}$ is connected to $Y$-axis symmetry,
$\mathcal{S}=-\tau \Theta$ to point symmetry around the origin,
and $\mathcal{MS}=-\mu_x \tau R^{1/2}\Theta$ to $X$-axis symmetry.

Based on  bifurcation patterns of $D_6$ 
described in section \ref{bifurcationOfD6},
bifurcations of 
figure-eight
solutions are summarized in 
the table \ref{tableBifurcationsFigureEight}.
The orbit of  solution bifurcated at $a=-0.2142$ under $U_h$
is shown in \fref{figSolutionBM02}.
All other orbits  are shown in \cite{Fukuda2019}.
Each bifurcation pattern yields each bifurcated solution
with different choreographic and geometrical symmetry.
The bifurcated solution by  bifurcation V is Sim\'o's H solution
\cite{SimoDynamical}.
The bifurcation V and VI was found by Mu\~{n}oz-Almaraz et al.~in 2006 \cite{Munoz2018}
and confirmed by the present authors \cite{Fukuda2019}.
Bifurcations I to IV are found by the present authors \cite{Fukuda2019}.
Numerical calculations show that
correspondence between parameter $\xi$ and $\kappa$
is one to two for I,
while one to one for II to VI bifurcations.
Therefore, 
order 1 bifurcation in I is surely fold
while in V is ``both-sides'', 
and order 2 bifurcations in II, III, IV and VI are ``one-side'' as expected.

\begin{table}
\caption{\label{tableBifurcationsFigureEight}
Bifurcation patterns of figure-eight solutions
characterized by $\mathcal{P_C}'$, $\mathcal{M}'$, and $\mathcal{S}'$
of the eigenfunction of Hessian $\mathcal{H}$.
The column $d$ represents degeneracy number or dimensions for the representation.
Symmetry represents the symmetry of the bifurcated solution.
The last column describes the type of bifurcation,
fold, one-side, or both-sides.
The bifurcated solution in I to IV are choreographic,
and V to VI are non-choreographic.
}
\footnotesize
\begin{tabular}{c|rrr|c|rrrl|l}
\br
Pattern & $\mathcal{P_C}'$ & $\mathcal{M}'$ & $\mathcal{S}'$ & $d$
&$a$ for $U_h$ &\multicolumn{2}{c}{$T$ for $U_{LJ}$}
& Symmetry &Type\\
&&&&&&$\alpha_-$&$\alpha_+$&\\
\hline
I & 1 & 1 & 1 & 1 &&14.479&14.479& $X$- and $Y$-axis&fold\\
II & 1 & 1 & $-1$ & 1 &&&17.132& $Y$-axis& one-side  \\
III & 1 & $-1$ & 1 & 1 &&&18.615& $O^{\rm b}$&one-side  \\
IV & 1 & $-1$ & $-1$ & 1 &$- 0.2142$&14.595&& $X$-axis&one-side  \\
\hline
V & 0 & 1 & $\pm 1$ & 2&0.9966&14.836&16.878& $X$- and $Y$-axis&both-sides \\
VI & 0 & $-1$ & $\pm 1$ & 2 &1.3424&14.861&16.111& $O^{\rm b}$ or$^{\rm a}$
$X$-axis&double$^{\rm a}$ one-side\\
\br
\end{tabular}

$^{\rm a}$ The  bifurcation VI yields 
two kind of bifurcated solutions
with different symmetry:
$O$ symmetric one
or $X$-axis symmetric one.\\
$^{\rm b}$ Symbol $O$ represents
point symmetry around the origin.
\end{table}
\begin{figure}
   \centering
\includegraphics[width=6cm]{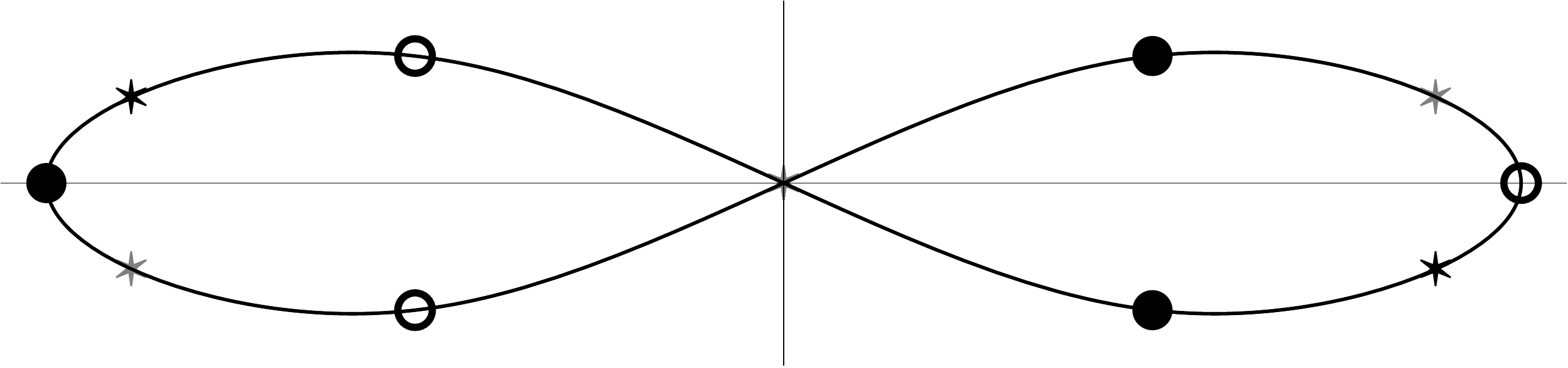}
\includegraphics[width=6cm]{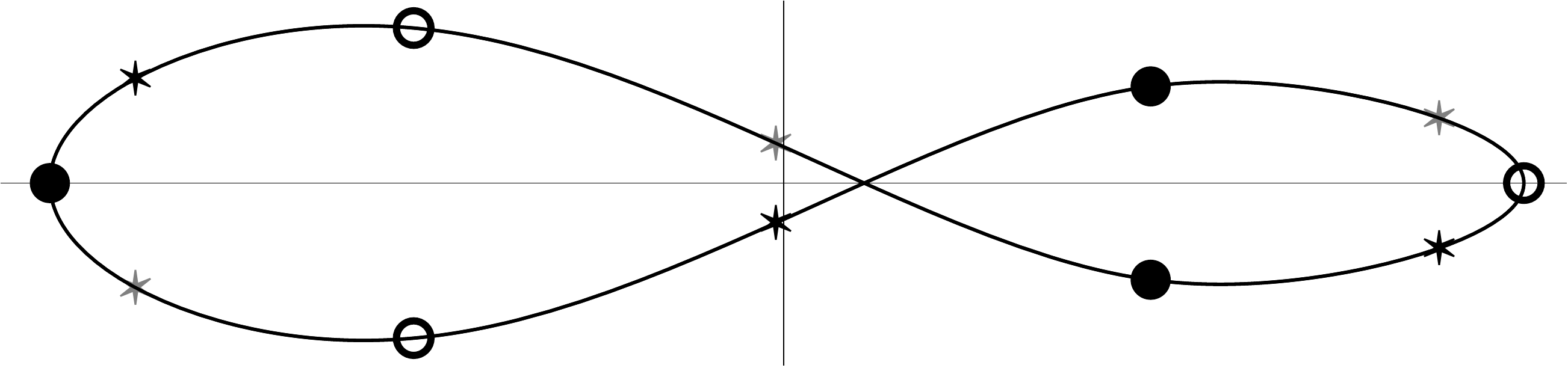}
   \caption{The figure-eight (left) and bifurcated solution for IV (right) 
   under $U_h$ at $a=-0.2$.
   Points represent position of particles at
   $t=-T/12$ (solid circles), $0$ (black stars), $T/12$ (hollow circles),
   and $2 T/12$ (grey stars).}
   \label{figSolutionBM02}
\end{figure}

Bifurcation in I is
fold
bifurcation between $\alpha_\pm$ solutions in $U_{LJ}$.

Bifurcations in II to IV 
are ``one-side'' bifurcations 
that yields ``less symmetric eights'',
namely choreographic solutions with less symmetry.
Existence of ``less symmetric eights'' were predicted by Alain Chenciner \cite{ChencinerICM2002,Chenciner2002}.
They surely exist in $U_h$ and $U_{LJ}$.
The bifurcation at $a=-0.2142$ in $U_h$
is a bifurcation of pattern IV
that
yields choreographic solution that 
looses $Y$-axis symmetry and keeps $X$-axis symmetry.
However, this branch does not reach $a=1$.
The orbit of the bifurcated solution at $a=-0.2$ is shown in
\fref{figSolutionBM02}.

Bifurcated solution in IV  has
projection operator $\mathcal{P}_{IV}=\mathcal{P_C}\mathcal{P_{MS}}$,
and one of the bifurcated solution  in VI
has $\mathcal{P}_{VI-}=\mathcal{P_{MS}}$.
They have non-vanishing angular momentum.
The reason is the following.
In these bifurcations both 
$\mathcal{M}$ symmetry and $\mathcal{S}$ symmetry are broken.
As a result,
the bifurcated solution  looses both 
$Y$-axis symmetry and the point symmetry around the origin.
Then the total signed area has non-vanishing value
which is equal to $T$ times of angular momentum $c$:
\begin{equation}
c T=\int_0^T d t\sum_{k=0,1,2}  r_k \times \frac{dr_k}{d t} \ne 0.
\end{equation}
Therefore, these bifurcated solutions have non-vanishing angular momentum.
See \fref{figSolutionBM02}.

Note that there are no direct paths $D_6 \to C_2=\{1,\mathcal{M}\}$ 
that would produce non-choreographic solution 
with $Y$-axis symmetry
and without $X$-axis symmetry. 
One possible 
path
is cascading bifurcation via Sim\'o's H:
$D_6 \to D_2\to C_2$. Another one is $D_6 \to C_6\to C_2$.
See \fref{figbifurcationsAndSymmetryBreaking}.

\section{Summary and discussions}
\label{summaryDiscussions}
We applied group theoretical method in bifurcation 
to investigate bifurcations of periodic solutions
in Lagrangian system. 
The results are summarized in the table 1, 2, 3 and the figure 1. 
In this method,
bifurcated solution is a stationary point of the action.
The second derivative of the action, Hessian $\mathcal{H}$, has 
important role.
A non-trivial zero of the eigenvalue of Hessian
yields bifurcation.
Since eigenvalues and eigenfunctions are classified by 
irreducible representations
of the symmetry of the Hessian,
group theories have important role in bifurcations.
In this method,
symmetry breaking pattern and symmetry of bifurcated solution
for each bifurcation is clear.
Symmetry of Lyapunov-Schmidt reduced action
apparently shows existence of the bifurcated solutions.
This method will give an alternative method to analyse bifurcations
of periodic solutions,
although this method will be mathematically equivalent to
methods based on Poincar\'{e} map or Floquet matrix.

\subsection{Bifurcations of  figure-eight solutions
and Sim\'{o}'s H}

This method gives theoretical explanations to 
numerically obtained bifurcations of figure-eight solutions
under $U_h$ with parameter $a$ 
and $U_{LJ}$ with parameter $T$ in the unified way.
The results are summarized in \tref{tableBifurcationsFigureEight}.

This method
also
predicts patterns for
bifurcations of Sim\'{o}'s H that has symmetry group $D_2$.
Since
$D_2$
is Abelian, it has 4  one-dimensional irreducible representations,
which are characterized by $\mathcal{M}'=\pm 1$ and $\mathcal{S}'=\pm 1$.
The bifurcation patterns are obvious
that are summarized in table \ref{table4RepsOfD2} and \fref{figbifurcationsAndSymmetryBreaking}.
\begin{table}
\caption{\label{table4RepsOfD2}
Four irreducible representations of Group $D_2$
characterized by $\mathcal{M}'$ and $\mathcal{S}'$.
All representations are one-dimensional ($d=1$).
The next two columns represent 
projection operators $\mathcal{P}$ and symmetry group $G$
for the bifurcated solution.
The order represents order $n$ of the bifurcation: 
$\kappa=-A_{n+2}\, r^n/(n+1)!+O(r^{n+1}),\, A_{n+2}\ne0$.
The last column describes the type of bifurcation.
}
\begin{indented}
\item[]
\begin{tabular}{c|rr|c|lccl}
\br
Representation&$\mathcal{M}'$ & $\mathcal{S}'$ &$d$& $\mathcal{P}$ &$G$&Order&Type\\
\hline
I$'$& 1 & 1 & 1 &$\mathcal{P_M}\mathcal{P_S}$&$D_2$&1&fold$^{\rm a}$\\
II$'$& 1 & $-1$ & 1 & $\mathcal{P_M}$&$C_2$ &2&one-side\\
III$'$& $-1$ & 1 & 1 &$\mathcal{P_S}$&$D_1$&2&one-side\\
IV$'$& $-1$ & $-1$ & 1 &$\mathcal{P_{MS}}$ &$D_1'$&2&one-side\\
\br
\end{tabular}
\item[]
$^{\rm a}$Fold bifurcation is suitable, although both-sides is still possible. 
\end{indented}
\end{table}

Gal\'an et al.~\cite{PhysRevLett.88.241101}
takes $m_0$ as a parameter.
In this case,
$m_0 \ne 1$ breaks $\sigma$ symmetry 
$\sigma(r_0,r_1,r_2)=(r_1,r_2,r_0)$,
while $\tau$ symmetry $\tau(r_0,r_1,r_2)=(r_0,r_2,r_2)$ is preserved.
Then, the symmetry group for the figure-eight solution is reduced
into $D_2$.
This is the  symmetry group for Sim\'{o}'s H solution.
Therefore, $D_2\to D_2$ bifurcation in \tref{table4RepsOfD2}
connects the figure-eight solution and Sim\'{o}'s H
by fold bifurcation \cite{PhysRevLett.88.241101}.

\subsection{Group theoretical bifurcation theory}
\subsubsection{Existence of at least one bifurcated solution
in each irreducible representation in $D_n$.}
The arguments in section \ref{secDegeneratedCase}
show that
at least one bifurcated solution exists
in each irreducible representations.
Actually, as shown above,
each irreducible representations of $D_6$ has
at least one projection operator that picks up one direction
of corollary \ref{corollaryProjectionOperator}.
Therefore, each irreducible representation has at least one 
bifurcated solution.
This is also true for $D_n$,
namely, there is at least one projection operator
such that $\mathcal{P}q_o=q_o$ and 
$\mathcal{P}\phi(\theta)=\phi(\theta)$
(one dimension) for each irreducible representation.
A proof is the following:
It is known that
the dimension of each irreducible representation of $D_n$
is one or two.
(For odd $n$:
2 one-dimensional representations and
$(n-1)/2$ two-dimensional ones.
For even $n$:
4 one-dimensional and $n/2-1$ two-dimensional ones.)
In two-dimensional representation,
the degeneracy comes from two eigenfunctions 
$\mathcal{S}\phi_\pm=\pm\phi_\pm$
with definition $\mathcal{S}q_o=q_o$.
Therefore the projection operator 
in corollary \ref{corollaryProjectionOperator}
that has the form
$\mathcal{P}=\mathcal{P}\mathcal{P_S}\ne 0$ always exists,
which picks up $\phi_+$ and excludes $\phi_-$.
This is the projection operator we are looking for.

\subsubsection{Similarity of bifurcation patterns.}
As shown in this paper,
the bifurcation patterns
depend only on the group structure of symmetry group $G$ 
for the original solution $q_o$ and the action $S$.
Therefore,
if two different systems have
symmetry groups $G$ and $G'$,
and $G$ and $G'$ are isomorphic or homomorphic,
the bifurcation patterns of the two systems are the same or similar.

For example. the bifurcation patterns
I to IV of $D_6$ in table \ref{table6RepsOfD6}
and I$'$ to IV$'$ of $D_2$ in table \ref{table4RepsOfD2}
is similar if we neglect $\mathcal{P_C}$ in column $\mathcal{P}$ in 
the former table.
The reason is the following.
In bifurcation patterns I to IV of $D_6$,
the symmetry of $\mathcal{C}$ is kept.
Since $\{1,\mathcal{C},\mathcal{C}^2\}$ is a
normal subgroup of $D_6$, we can make quotient group:
\begin{equation}
D_6/\{1,\mathcal{C},\mathcal{C}^2\}
\cong \{1, \mathcal{M},\mathcal{S},\mathcal{S}\mathcal{M}\}
=D_2.
\end{equation}
Therefore,
bifurcations of $D_6$ keeping $\mathcal{C}$ symmetry
are equivalent to bifurcations of $D_2$.

In the next sub-subsection,
we consider cases where
two completely different system having isomorphic symmetry groups.

\subsubsection{Period $k$ bifurcation.}
\label{BifurcationDn}

Consider period $k$ bifurcations of
a figure-eight solution.
For this case, the periodic condition for an variation
$\delta q(t)$ should be
\begin{equation}
\delta q(t+k T)=\delta q(t),
\end{equation}
where $T$ is the period of 
this figure-eight solution.
That means $\mathcal{B}^{6 k}=\mathcal{R}^k=1$, instead of 
$\mathcal{B}^6=\mathcal{R}=1$.
Therefore, the symmetry group is
\begin{equation}
D_{6 k}=\{1,\mathcal{B},\dots,\mathcal{B}^{6 k-1},
\mathcal{S}, \mathcal{S}\mathcal{B},\dots,
\mathcal{S}\mathcal{B}^{6 k-1}\},\,
\mathcal{B}\mathcal{S}=\mathcal{S}\mathcal{B}^{-1}.
\end{equation}
For example,  
period $k=5$ bifurcation will be determined by 
irreducible representations of $D_{30}$. 
Some of $k=5$ slalom solutions 
by M.~\v{S}uvakov and V.~Dmitra\v{s}inovi\'{c} \cite{Suvakov2013,Suvakov2013NumericalSearch},
and M.~\v{S}uvakov and M.~Shibayama \cite{SuvakovShibayama2016}
will turn out to be bifurcated solutions
of the figure-eight by period 5 bifurcation.

Similarly,
period $3$ bifurcations of Sim\'o's H ($D_2=\{1,\mathcal{M},\mathcal{S},\mathcal{S}\mathcal{M}\}$)
will be described as bifurcation of 
\begin{equation}
D_6'=\{1,\mathcal{M},\mathcal{M}^2,\dots,\mathcal{M}^5,
	\mathcal{S},\mathcal{S}\mathcal{M},\mathcal{S}\mathcal{M}^2,\dots,
	\mathcal{S}\mathcal{M}^5\}
\cong D_6,
\end{equation}
because $\mathcal{M}^6=\mathcal{R}^3=1$ for period 3 bifurcation.
Namely, it must have the same bifurcation patterns
in \tref{table6RepsOfD6} and \fref{figbifurcationsAndSymmetryBreaking}.
Moreover,
period $2$ of $D_3$ or period $6$ of $D_1$
bifurcation will be described by $D_6$.

Consider a periodic solution that is invariant
under an operator $\mathcal{S}=\mathcal{O}\Theta$
where $\Theta$ is the time reversal and $\mathcal{O}$
satisfies $(\mathcal{O}\Theta)^2=1$, 
$\mathcal{O}\mathcal{R}=\mathcal{R}\mathcal{O}$.
A simple example for $\mathcal{O}$ is $\mathcal{O}=-1$.
In this case, $\mathcal{S}q_o=q_o$ means $q_o(-t)=-q_o(t)$.
In general, assuming $q_o$ has no other invariance,
the symmetry group for $q_o$ and $S$ is $D_1=\{1,\mathcal{S}\}$.
Then, period $k$ bifurcation of this solution
will be described by the  dihedral group of 
regular $k$-gon,
\begin{equation}
D_k'=\{1,\mathcal{R},  \dots, \mathcal{R}^{k-1},
	\mathcal{S}, \mathcal{S}\mathcal{R}, \dots,\mathcal{S}\mathcal{R}^{k-1}
	\},\,
	\mathcal{R}\mathcal{S}=\mathcal{S}\mathcal{R}^{-1}.
\end{equation}
For example,
period doubling bifurcation of this system
should be described by bifurcation of $D_2$,
and period 3 bifurcation  by  $D_3$,
and period 6 bifurcation by  $D_6$.
Indeed,
in the  book of K.~R.~Meyer 
and G.~R.~Hall \cite{MeyerHall}
or Meyer and D.~C.~Offin \cite{MeyerOffin},
they describe period doubling,
period 3
and period 6 bifurcations
for Hamiltonian system
that are exactly expected for bifurcations in 
faithful representation of 
$D_2$ (pattern IV or IV$'$ in this paper),
$D_3$ (pattern V)
and $D_6$ (pattern VI),
although we don't consider the stability
of original and bifurcated solution(s) here.
See sections
``Period doubling'' and  ``$k$-bifurcation points''
in \cite{MeyerHall} or \cite{MeyerOffin}.

For $k=2$, $D_2'$ is
\begin{equation}
D_2'=\{1,\mathcal{R},\mathcal{S},\mathcal{SR}\}.
\end{equation}
The faithful representation is $\tilde\mathcal{R}=\tilde\mathcal{S}=-1$.
Therefore bifurcation pattern is
\begin{equation}
D_2' \to \{1, \mathcal{R S}\}=D_1''.
\end{equation}
Therefore,
period doubling bifurcation should be order 2 bifurcation
with  bifurcated solution that satisfies $\mathcal{R S}q_b=q_b$
on one side of parameter.
Two solutions $\{q_b, \mathcal{S}q_b=\mathcal{R}q_b\}$
exist at one side of the bifurcation point.

The faithful representation of $D_3$ is
\begin{equation}
\mathcal{R}=\left(\begin{array}{rr}
	\cos(2\pi/3) & -\sin(2\pi/3) \\
	\sin(2\pi/3) & \cos(2\pi/3)
\end{array}\right),\,
\mathcal{S}=\left(\begin{array}{cc}
	1 & 0 \\
	0 & -1
\end{array}\right).
\end{equation}
This is the same as the irreducible representation in pattern V.
The symmetry breaking pattern in this bifurcation is
\begin{equation}
D_3''=\{1,\mathcal{R},\mathcal{R}^2,
		\mathcal{S},\mathcal{S}\mathcal{R},\mathcal{S}\mathcal{R}^2
	\}
\to
D_1=\{1,\mathcal{S}\}.
\end{equation}
Therefore,
period 3 bifurcation should  be order 1 with $\mathcal{S}q_b=q_b$.
Three bifurcated solutions $\{q_b,\mathcal{R}q_b,\mathcal{R}^2 q_b\}$
exist for both side of parameter.

The faithful representations of $D_6$ produce 
$D_6 \to D_1 \mbox{ or } D_1'$ as shown in
bifurcation pattern VI in this paper.
Therefore
period 6 bifurcation should be order 2
with two kinds of bifurcated solutions
for one side of parameter.
One satisfies $\mathcal{S}q_{1}=q_{1}$,
and other satisfies $\mathcal{R}^3\mathcal{S}q_{2}=q_{2}$.
Each of them has 6 copies
$\{q_1,\mathcal{R}q_1,\dots,\mathcal{R}^5 q_1\}$
and $\{q_2,\mathcal{R}q_2,\dots,\mathcal{R}^5 q_2\}$.

\subsubsection{Symmetry breaking of bifurcation
and preserving of the action.}
As shown in \sref{bifurcationOfD6},
bifurcation in II to VI breaks a symmetry or symmetries.
While, symmetry of action $S$ is always preserved.
The Lyapunov-Schmidt reduced action $S_{LS}$ inherits this invariance.  
As a result, multiple copies of a bifurcated solution by the broken symmetry
will emerge from the bifurcation point:
$\{q_b, g q_b,g^2 q_b,\dots,g^{n_g-1}q_b\}$
where $g$ is a broken symmetry and $n_g$ is the order of $g$.
The locus of copies are congruent.

For example, in bifurcation V,
the bifurcated solution $q_b$ breaks $\mathcal{C}$ invariance,
therefore the invariance of the action under the transformation $\mathcal{C}$
yields three solutions $\{q_b, \mathcal{C}q_b, \mathcal{C}^2 q_b\}$
that have congruent locus.
Similarly,
the breaking of $\mathcal{M}$ yields
two bifurcated solutions $\{q_b, \mathcal{M}q_b\}$.
It is the same for $\mathcal{S}$.
Note that
in the case both $\mathcal{M}$ and $\mathcal{S}$ are broken
while $\mathcal{MS}$ is preserved,
we still have two bifurcated solutions 
$\{q_b, \mathcal{S}q_b=\mathcal{M}q_b\}$
with congruent locus,
since the
invariance $\mathcal{MS}q_b=q_b$
ensures $\mathcal{S}q_b=\mathcal{M}q_b$.
Thus bifurcations in II to IV yield
two bifurcated solutions with congruent locus.
Similarly,
bifurcation  VI yields solutions in
two congruent classes
$\{q_{b+}, \mathcal{B}q_{b+},\dots,\mathcal{B}^5 q_{b+}\}$
and
$\{q_{b-}, \mathcal{B}q_{b-},\dots,\mathcal{B}^5 q_{b-}\}$.

On the other hand, any unbroken symmetry is invisible.
This is because if $g\phi(\theta)=\phi(\theta)$,
the symmetry for the reduced action yields identity
$S_{LS}(r,\theta)=S_{LS}(r,\theta)$ that has no information.
Similarly, if $g q_b=q_b$, $g$ does not produce new solution.

\subsection{Further investigations}
\subsubsection{Stability.}
\label{secStability}
Until now, relations between the behaviour of the action around a solution and the stability of the solution is unclear.
For this reason,
we used terms ``both-sides'' or ``one-side''
instead of ``trans-critical'' or ``pitchfork''.
Actually,
the bifurcation at $a=0.9966$ and $a=1.3424$ for $U_h$ 
does not change the stability of the figure-eight solution
\cite{Munoz2018}.
We confirmed their results.
The stability change/unchanged at $a=-0.2142$ still needs careful investigations.
As shown in \ref{BifurcationDn},
bifurcations of the figure-eight solution
at $a=-0.2142,\, 0.9966$ and $1.3424$
is equivalent, in a group theoretical point of view, 
to period 2, 3, and 6 bifurcations.
We suspect this might be an origin of exotic behaviour of stability change
at these points.
So,
further systematic and careful numerical investigations 
and theoretical developments for changing stability 
at  bifurcation points
with a group theoretical point of view
are required.

\subsubsection{Stationary point at finite distance.}
To describe bifurcations,
it is enough to consider infinitesimally small distance
from the original periodic solution.
Then the first non-zero $A_n$ in \eref{eqBifurcatedSolutionForR}
determines the properties of each bifurcation.
Can we predict the existence of other periodic solutions
at finite distance
from the derivatives at original?

To make the argument clear,
take the figure-eight solution as an original one
and consider the subspace selected by projection operator
$\mathcal{P}=\mathcal{P_M}\mathcal{P_S}$,
which is the subspace where Sim\'{o}'s H solution lives.
The reduced action  is 
a function of one variable $r$:
\begin{equation}
S_{LS}(r)
=\frac{\kappa}{2}r^2+\sum_{n\ge 3}\frac{A_n}{n!}r^n
=\frac{\kappa}{2}r^2
	+\frac{A_3}{3!}r^3
	+\frac{A_4}{4!}r^4
	+\dots.
\end{equation}
The term $A_3 r^3/3!$ goes to $-\infty$ 
for either $r\to +\infty$ or $-\infty$.
Now, consider what will happen if $A_4$ is positive
or simply if $S_{LS}(r)$ goes to sufficiently large positive for a finite value of $r$.
Then, there must be at least one more stationary point at a finite distance.
By the theorem \ref{theoremProjectionOperator},
any stationary point in this subspace is a solution of the equations of motion
and the solution has the symmetry of Sim\'{o}'s H:
$\mathcal{P}=\mathcal{P_M}\mathcal{P_S}$.
Such solution surely exists near the figure-eight and Sim\'{o}'s H solutions.
Mun\~oz-Almaraz et al.~\cite{Munoz2018} showed numerically that
Sim\'{o}'s H solution has fold bifurcations at the both end of
a interval
of $a$. Therefore, near the end of
this interval,
the figure-eight, 
Sim\'{o}'s H and one other solution exist.
Then at the end of the interval,  Sim\'{o}'s H and the other solution
are pair-annihilated by fold bifurcation.
The present authors confirmed their results numerically.

Similar question  also arises for bifurcation pattern VI.
Numerical calculations for $U_h$ show that the bifurcated solutions emerge
in $\kappa>0$ side.
This means $A_4(0)<0$.
Then if $A_{6\pm}>0$ or if $S_{LS}$ becomes sufficiently large positive
for a finite value of $r$,
then at least another solution exists at finite distance.

Can we theoretically treat the figure-eight solution,
bifurcated solution(s)
and solution(s) at finite distance
at once,
and can describe the observed fold bifurcations?
It will be very interesting
if we can do it
by  considering the behaviour of action around
the figure-eight solution.

\subsubsection{Equality  of the second derivatives of $S_{LS}(r,\theta)$
and the eigenvalues of Hessian at $q_b$.}
We have shown in appendix D that the second derivative of the reduced action
at $q_b$
is equal to corresponding eigenvalues of Hessian
at $q_b$ to order $r^2$
using ordinary perturbation method for the eigenvalues.
However,
the interesting term in \eref{VIpatialThetaSquareS}
is of order $r^4$.
So, equality to order $r^2$ is not enough for bifurcation VI.
Since ordinary perturbation methods are tedious and inefficient,
we need to find efficient methods to show
equality or inequality of them.

\ack
The present authors express sincere thanks 
to Hiroshi Kokubu, Mitsuru Shibayama, 
Kazuyuki Yagasaki,
and Michio Yamada
for valuable discussions and comments for this work.
Special thanks to 
Francisco Javier Mu\~{n}oz-Almaraz,
Jorge Gal\'{a}n-Vioque, Emilio Freire and Andre Vanderbauwhede
for give us very valuable notes.
This work was supported by JSPS Grant-in-Aid for Scientific Research 17K05146 (HF) and 17K05588 (HO).

\appendix

\section{Proof of inheritance of symmetry}
\label{secInvarianceOfIntegrals}
In this section,
we will prove the equalities  in section \ref{secInheritance}.

A proof of  lemma \ref{lemmaInvariance} is the followings;
\begin{proof}
Since $G$ is a symmetry group for $q_o$ and $S$,
for arbitrary variation $\delta q$, we have
\begin{equation}
S[q_o+\delta q]=S[q_o + g \delta q].
\end{equation}
Expansions of the action around $q_o$ yields
\begin{equation}
\fl
\frac{1}{2}\int d t\, \delta q \mathcal{H}\delta q 
	+ \sum_{n\ge 3}\frac{1}{n!}\braket{\delta q^n}
=\frac{1}{2}\int d t\, (g\delta q) \mathcal{H}(g\delta q) 
	+ \sum_{n\ge 3}\frac{1}{n!}\braket{(g\delta q)^n}.
\end{equation}
Since, $\delta q$ is arbitrary function, this equation holds for order by order:
\begin{eqnarray}
\int d t\, \delta q \mathcal{H}\delta q 
=\int d t\, (g\delta q) \mathcal{H}(g\delta q),\\
\braket{\delta q^n}=\braket{(g\delta q)^n}.\label{nThOrderTerm}
\end{eqnarray}
So, if we take $\delta q = a_i e_i + a_j e_j$, we get
\begin{equation}
\int d t\, (g e_i)\mathcal{H}(g e_j)=\int d t\, e_i\mathcal{H}e_j=\lambda_i \delta_{ij}.
\end{equation}
Where $\lambda_i$ is the eigenvalue of $\mathcal{H}$ for $e_i$.
Moreover, if we 
put $\delta q=a_i e_i+a_j e_j +\dots$ into \eref{nThOrderTerm}
and
compare the corresponding term $a_i\vphantom{)}^{n_i}a_j\vphantom{)}^{n_j}\dots$
of left and right side,
we get the lemma \ref{lemmaInvariance}.
\end{proof}

Next,
let us prove the theorem \ref{theoryInvarianceOfSLS}.
\begin{proof}
The definition of $S_{LS}(r,\theta)$ and the invariance of action under $g$
yields
\begin{eqnarray}
S_{LS}(r,\theta)
&=S\left[q_o+r\phi(\theta)+r\sum_\alpha \epsilon_\alpha(r,\theta) \psi_\alpha)
	\right]\\
&=S\left[q_o+r'\phi(\theta')+r\sum_\alpha \epsilon_\alpha(r,\theta) g\psi_\alpha)
	\right].
\end{eqnarray}
So, if
\begin{equation}
\label{epsilonForThetaDash}
r\epsilon_\alpha(r,\theta)g\psi_\alpha
=r'\epsilon_\alpha(r',\theta')\psi_\alpha
\end{equation}
is satisfied, \eref{invarianceSLS} is satisfied.
Where $r'\epsilon_\alpha(r',\theta')$ is the solution of \eref{recEqForEpsilon}
for $r'\phi(\theta')$.
Therefore our goal is to show that $r'\epsilon_\alpha(r',\theta')$ in \eref{epsilonForThetaDash}
surely satisfies \eref{recEqForEpsilon} for $r'\phi(\theta')$.
This is true,
because $\epsilon_\alpha$ in \eref{recEqForEpsilon} is the unique solution
for $r$ and $\phi(\theta)$ including the sign.
However, a direct proof will be interesting.

Now, let us show this.
For each $\psi_\alpha$, there is an orthogonal matrix representation 
$\tilde{g}_\alpha$ of $g$,
\begin{equation}
g\psi_\alpha
=\psi_\alpha\tilde{g}_\alpha.
\end{equation}
If
$\psi_\alpha$
belongs one-dimensional representation,
$\tilde{g}_\alpha=\pm1$.
If 
$\psi_\alpha$
belongs two-dimensional representation,
$\tilde{g}_\alpha$ is a 2 by 2 matrix:
for $s=\pm$,
\begin{equation}
g\psi_{\alpha s}=\sum_{s'=\pm1}\psi_{\alpha s'}\tilde{g}_{\alpha s', \alpha s},
\mbox{ and }
\sum_{s=\pm}\tilde{g}_{\alpha s', \alpha s}\tilde{g}_{\alpha s'', \alpha s}
=\delta_{s',s''}.
\end{equation}

Let us start from the definition \eref{recEqForEpsilon} of
$\epsilon_\alpha(r,\theta)$.
Then the invariance of integrals under $g$ yields
\begin{eqnarray}
\fl
r\epsilon_{\alpha s}(r,\theta)
=-\frac{1}{\lambda_\alpha}
	\sum_{n\ge 3}\frac{1}{(n-1)!}
		\Braket{
			\left(r g\phi(\theta)
				+r\sum_\beta\epsilon_{\beta}(r,\theta)g\psi_\beta\right)^{n-1}
				g\psi_{\alpha s}
				}\nonumber\\
\fl
=-\frac{1}{\lambda_\alpha}
	\sum_{n\ge 3}\frac{1}{(n-1)!}
		\Braket{
			\left(r'\phi(\theta')
			+r\sum_\beta
				\psi_\beta
				\tilde{g}_\beta
				\epsilon_\beta(r,\theta)\right)^{n-1}
				\sum_{s'=\pm}
					\psi_{\alpha s'}\tilde{g}_{\alpha s', \alpha s}
				},
\end{eqnarray}
where $\psi_\beta \tilde{g}\epsilon_\beta$ is the matrix  notation for
\begin{equation}
\psi_\beta \tilde{g}_\beta\epsilon_\beta
=\sum_{s',s''=\pm}\psi_{\beta s'} \tilde{g}_{\beta s', \beta s''}
						\epsilon_{\beta s''}.
\end{equation}
Therefore,
\begin{eqnarray}
\fl
r(\tilde{g}_\alpha\epsilon_\alpha(r,\theta))_{s''}
=r\sum_{s=\pm} \tilde{g}_{\alpha s'',\alpha s}\epsilon(r,\theta)_{\alpha s}
\nonumber\\
\fl
=-\frac{1}{\lambda_\alpha}\sum_{n\ge 3}\frac{1}{(n-1)!}
	\Braket{
			\left(r'\phi(\theta')
			+r\sum_\beta
				\psi_\beta
				\tilde{g}_\beta
				\epsilon_\beta(r,\theta)
			\right)^{n-1}
				\sum_{s}\tilde{g}_{\alpha s'',\alpha s}
				\sum_{s'}
					\psi_{\alpha s'}
					\tilde{g}_{\alpha s',\alpha s}
			}\nonumber\\
\fl
=-\frac{1}{\lambda_\alpha}\sum_{n\ge 3}\frac{1}{(n-1)!}
	\Braket{
			\left(
			  r'\phi(\theta')
			+r\sum_\beta
				\psi_\beta
				\tilde{g}_\beta
				\epsilon_\beta(r,\theta)
			\right)^{n-1}
				\psi_{\alpha s''}
	}.
\end{eqnarray}
So,
$r\tilde{g}_\alpha\epsilon_\alpha(r,\theta)$ satisfies
the definition for $r'\epsilon_\alpha(r',\theta')$.
This is what we wanted to show.
\end{proof}

\section{$A_k(\theta)$, $k=3,4$ for V}
\label{AkForV}
In this section, we calculate $A_k(\theta)$, $k=3,4$ for bifurcation V.
We use notation $\mathcal{S}\phi_\pm=\pm\phi_\pm$
and $\phi(\theta)=\cos(\theta)\phi_++\sin(\theta)\phi_-$.
We don't need $\mathcal{M}$ operator for calculations
in this section.

\subsection{$A_3(\theta)$ for V}
\label{secA3forV}
Expansion of $\phi(\theta)^3$ yields
\begin{equation}
\label{phi3ForVstep1}
\braket{\phi(\theta)^3}
=\cos(\theta)^3\braket{\phi_+^3}+3\cos(\theta)\sin(\theta)^2\braket{\phi_+\phi_-^2},
\end{equation}
because $\braket{\phi_+^2\phi_-}=\braket{\phi_-^3}=0$ 
by $\mathcal{S}\phi_\pm=\pm\phi_\pm$.
Using the invariance of the lemma \ref{lemmaInvariance}
for  $g=\mathcal{B}$,
we have $\braket{\phi_+^3}=\braket{(\mathcal{B}\phi_+)^3}$
and $\braket{\phi_+\phi_-^2}=\braket{(\mathcal{B}\phi_+)(\mathcal{B}\phi_-)^2}$.
On the other hand,
$\phi_\pm$ are mixed by $\mathcal{B}$.
Using the expression $\tilde\mathcal{B}$,
\begin{equation}
(\mathcal{B}\phi_+ ,\mathcal{B}\phi_-)
=(\phi_+,\phi_-)\tilde\mathcal{B}
=(\phi_+,\phi_-)
\left(\begin{array}{rr}
	-1/2& -\sqrt{3}/2\\
	\sqrt{3}/2&-1/2
\end{array}\right)
\end{equation}
Then, we have
\begin{equation}
\fl
\braket{\phi_+^3}
=\braket{(\mathcal{B}\phi_+)^3}
=\Braket{\left(-\frac{1}{2}\phi_+ + \frac{\sqrt{3}}{2}\phi_-\right)^3}
=-\frac{1}{8}\braket{\phi_+^3}-\frac{9}{8}\braket{\phi_+\phi_-^2}.
\end{equation}
Here we have used $\braket{\phi_+^2\phi_-}=\braket{\phi_-^3}=0$
again.
Similar equation holds for
$\braket{\phi_+\phi_-^2}=\braket{(\mathcal{B}\phi_+)(\mathcal{B}\phi_-)^2}$.
Assembling two equations,
we get the following equation:
\begin{equation}
\fl
\left(\begin{array}{l}\braket{\phi_+^3} \\\braket{\phi_+\phi_-^2}\end{array}\right)
=\left(\begin{array}{l}
	\braket{(\mathcal{B}\phi_+)^3} \\
	\braket{(\mathcal{B}\phi_+)(\mathcal{B}\phi_-)^2}
\end{array}\right)
=\left(\begin{array}{rr}
	-1/8 & -9/8 \\
	-3/8 & 5/8
\end{array}\right)
\left(\begin{array}{l}\braket{\phi_+^3} \\\braket{\phi_+\phi_-^2}\end{array}\right).
\end{equation}
Namely, $(\braket{\phi_+^3}, \braket{\phi_+\phi_-^2})$ must be an eigenvector
of the matrix in the right  side for eigenvalue 1.
The solution is 
\begin{equation}
\label{phiPlusCubeV}
\left(\begin{array}{l}
\braket{\phi_+^3} \\
\braket{\phi_+\phi_-^2}
\end{array}\right)
=\left(\begin{array}{r}1 \\-1\end{array}\right)\braket{\phi_+^3}.
\end{equation}
Substituting this solution into \eref{phi3ForVstep1},
we get
\begin{equation}
\braket{\phi(\theta)^3}
=\left(\cos(\theta)^3-3\cos(\theta)\sin(\theta)^2\right)\braket{\phi_+^3}
=\cos(3\theta)\braket{\phi_+^3}.
\end{equation}

\subsection{$A_4(\theta)$ for V}
\label{secA4forV}
The invariance of integrals by $\phi_\pm \to \mathcal{B}\phi_\pm$
yields
\begin{equation}
\label{phi4VComponents}
\fl
{\left(\begin{array}{l}
\braket{\phi_+^4} \\
\braket{\phi_+^2\phi_-^2} \\
\braket{\phi_-^4}
\end{array}\right)}
={\left(\begin{array}{rrr}
1/16 & 9/8 & 9/16 \\
3/16 & -1/8 & 3/16 \\
9/16 & 9/8 & 1/16
\end{array}\right)}%
{\left(\begin{array}{l}
\braket{\phi_+^4} \\
\braket{\phi_+^2\phi_-^2} \\
\braket{\phi_-}^4
\end{array}\right)
}
={\left({\begin{array}{c}1 \\1/3 \\1\end{array}}\right)}
\braket{\phi_+^4}.
\end{equation}
\begin{equation}
\label{phi4V}
\fl
\therefore
\braket{\phi(\theta)^4}
=\cos(\theta)^4\braket{\phi_+^4}
	+6\cos(\theta)^4\sin(\theta)^2\braket{\phi_+^2\phi_-^2}
	+\sin(\theta)^4\braket{\phi_-^4}
=\braket{\phi_+^4}.
\end{equation}

Now, we proceed to calculate 
$\sum_\alpha \lambda_\alpha^{-1}\braket{\phi(\theta)^2\psi_\alpha}^2$;
\begin{equation}
\label{phiSquarePsiV}
\fl
\mbox{For }
\mathcal{P_C}\psi_{\alpha+}=\psi_{\alpha+}
\mbox{: }
\left(\begin{array}{c}
\braket{\phi_+^2\psi_{\alpha_+}}\\
\braket{\phi_-^2\psi_{\alpha_+}}
\end{array}\right)
=\left(\begin{array}{cc}
1/4 & 3/4 \\
3/4 & 1/4
\end{array}\right)
\left(\begin{array}{c}
\braket{\phi_+^2\psi_{\alpha_+}}\\
\braket{\phi_-^2\psi_{\alpha_+}}
\end{array}\right)
=\left(\begin{array}{c}1 \\1\end{array}\right)
	\braket{\phi_+^2\psi_{\alpha_+}}.
\end{equation}
\begin{equation}
\therefore
\braket{\phi(\theta)^2 \psi_{\alpha+}}
=\cos(\theta)^2\braket{\phi_+^2\psi_+}+\sin(\theta)^2\braket{\phi_-^2\psi_+}
=\braket{\phi_+^2\psi_{\alpha+}}.
\end{equation}
\begin{equation}
\fl
\mbox{For }\mathcal{P_C}\psi_{\beta-}=\psi_{\beta-}
\mbox{: }
\braket{\phi_+\phi_-\psi_{\beta-}}
=-\frac{1}{2}\braket{\phi_+\phi_-\psi_{\beta-}}
=0.
\end{equation}
\begin{equation}
\label{phiPlusPhiMinusPsiMinusV}
\therefore
\braket{\phi(\theta)^2\psi_{\beta-}}
=2\cos(\theta)\sin(\theta)\braket{\phi_+\phi_-\psi_{\beta-}}
=0.
\end{equation}

\begin{eqnarray}
\label{phiPlusSquarePsiMinusV}
\fl
\mbox{For }\mathcal{P_C}\psi_{\gamma\pm}=0\mbox{: }
\left(\begin{array}{l}
\braket{\phi_+^2\psi_{\gamma+}}\\
\braket{\phi_-^2\psi_{\gamma+}} \\
\braket{\phi_+\phi_-\psi_{\gamma-}}
\end{array}\right)
&=
\left(\begin{array}{rrr}
-1/8 & -3/8 & \mp 3/4 \\
-3/4 & -1/8 & \pm 3/4 \\
\mp 3/8 & \pm 3/8 & 1/4
\end{array}\right)
\left(\begin{array}{l}
\braket{\phi_+^2\psi_{\gamma+}}\\
\braket{\phi_-^2\psi_{\gamma+}} \\
\braket{\phi_+\psi_-\psi_{\gamma-}}
\end{array}\right) \nonumber\\
&=\left(\begin{array}{r}1 \\-1 \\-1\end{array}\right)
	\braket{\phi_+^2\psi_{\gamma+}},
\end{eqnarray}
where
upper and lower sign represent
the sign for $\psi_\gamma$ in representation V 
and in VI respectively.
Therefore,
\begin{eqnarray}
\fl
\braket{\phi(\theta)^2\psi_{\gamma+}}
=\cos(\theta)^2\braket{\phi_+^2\psi_{\gamma+}}
	+\sin(\theta)^2\braket{\phi_-^2\psi_{\gamma+}}
=\cos(2\theta)\braket{\phi_+^2\psi_{\gamma+}},\\
\fl
\braket{\phi(\theta)^2\psi_{\gamma-}}
=2\cos(\theta)\sin(\theta)\braket{\phi_+\phi_-\psi_{\gamma-}}
=-\sin(2\theta)\braket{\phi_+^2\psi_{\gamma+}},
\end{eqnarray}
\begin{equation}
\fl
\sum_{\mathcal{P_C}\psi_{\gamma}=0}
	\frac{1}{\lambda_\gamma}
	\left(\braket{\phi(\theta)^2\psi_{\gamma+}}^2
		+\braket{\phi(\theta)^2\psi_{\gamma-}}^2
		\right)
=\sum_{\mathcal{P_C}\psi_{\gamma}=0}
	\frac{1}{\lambda_\gamma}
	\braket{\phi_+^2\psi_{\gamma+}}^2.
\end{equation}
Assembling these terms, we get
\begin{eqnarray}
\sum_\alpha \frac{1}{\lambda_\alpha}
	\braket{\phi(\theta)^2\psi_\alpha}^2 
=\sum_{\alpha}
	\frac{1}{\lambda_\alpha}
	\braket{\phi_+^2\psi_{\alpha+}}^2.
\end{eqnarray}

So, we get, 
\begin{eqnarray}
\label{A4forV}
A_4(\theta)
=\braket{\phi_+^4}
	-\sum_{\alpha}
		\frac{3}{\lambda_\alpha}\braket{\phi_+^2\psi_{\alpha+}}^2
=A_4(0).
\end{eqnarray}

\section{$A_k(\theta)$, $k=3,4,5,6$ for VI}
\label{AkForVI}
In this section, we calculate $A_k(\theta)$, $k=3,4,5,6$ for bifurcation VI.

\subsection{$A_3(\theta), A_5(\theta), \dots$ for VI}
By $\mathcal{M}\phi(\theta)=-\phi(\theta)$
and the theorem \ref{theoryInvarianceOfSLS},
The reduced action is an even function of $r$:
$S_{LS}(r,\theta)=S_{LS}(-r,\theta)$.
Therefore all $A_{2 n+1}(r,\theta)=0$ for $n=1,2,\dots$.
However, direct check will be interesting.

$A_3(\theta)=\braket{\phi(\theta)^3}=0$
is obvious by $\mathcal{M}\phi(\theta)=-\phi(\theta)$.

Three terms contribute to $A_5(\theta)$:
\begin{equation}
\fl
\braket{\phi(\theta)^5},\,
\sum_\alpha
	\frac{
		\braket{\phi(\theta)^2\psi_\alpha}
		\braket{\phi(\theta)^3\psi_\alpha}
	}{\lambda_\alpha},\,
\sum_{\alpha,\beta}
	\frac{
		\braket{\phi(\theta)^2\psi_\alpha}
		\braket{\phi(\theta)\psi_\alpha\psi_\beta}
		\braket{\phi(\theta)^2\psi_\beta}
	}{\lambda_\alpha \lambda_\beta}.
\end{equation}
The first term is zero, because $\mathcal{M}\phi(\theta)=-\phi(\theta)$.
The second term is zero, 
because
$\braket{\phi^2\psi_\alpha}=0$ 
if $\mathcal{M}\psi_\alpha=-\psi_\alpha$,
and 
$\braket{\phi^3\psi_\alpha}=0$ 
if $\mathcal{M}\psi_\alpha=\psi_\alpha$.
For the last term,
$\psi_\alpha$ and $\psi_\beta$ should belongs to
$\mathcal{M}'=1$ to give non-zero values to
$\braket{\phi(\theta)^2\psi_\alpha}$ and 
$\braket{\phi(\theta)^2\psi_\beta}$.
However, it gives
$\braket{\phi(\theta)\psi_\alpha\psi_\beta}=0$.
Therefore, the last term is also zero.
So we get
$A_5(r,\theta)=0$.

Therefore $\mathcal{M}\phi(\theta)=-\phi(\theta)$
surely ensure $A_3(\theta)=A_5(\theta)=0$
as predicted by the theorem \ref{theoryInvarianceOfSLS}.

\subsection{$A_4(\theta)$ for VI}
The term $A_4(\theta)$ has the same expression as in \eref{A4forV}.
\begin{eqnarray}
\label{A4forVI}
A_4(\theta)
=\braket{\phi_+^4}
	-\sum_{\alpha}
		\frac{3}{\lambda_\alpha}\braket{\phi_+^2\psi_{\alpha+}}^2
=A_4(0).
\end{eqnarray}
Because we can use $\mathcal{C}$ to calculate the relations between
$\braket{\phi_+^4}$, $\braket{\phi_+^2\phi_-^2}$ and $\braket{\phi_-^4}$,
etc.,
for example $\braket{\phi_+^4}=\braket{(\mathcal{C}\phi_+)^4}$.
Since the representation of $\mathcal{C}$ in VI is the same as that
of $\mathcal{B}$ in V,
we get the same relations in V.

\subsection{$A_6(\theta)$ for VI}
\label{secA6forVI}
Seven terms contribute to $A_6(\theta)$:
Here we pick up two simpler terms 
$\braket{\phi(\theta)^6}$ and 
$\sum \braket{\phi(\theta)^3\psi_\alpha}^2/\lambda_\alpha$.
The invariance of  integrals under $\phi_\pm \to \mathcal{C}\phi_\pm$
yields
\begin{equation}
\left(\begin{array}{l}
\braket{\phi_+^6} \\
\braket{\phi_+^4\phi_-^2} \\
\braket{\phi_+^2\phi_-^4}\\
\braket{\phi_-^6}
\end{array}\right)
=
\frac{1}{64}\left(\begin{array}{rrrr}
1 & 45 & 135 & 27\\
3 & 31 & -27 & 9 \\
9 & -27 & 31 & 3 \\
27 & 135 & 45 & 1
\end{array}\right)
\left(\begin{array}{l}
\braket{\phi_+^6} \\
\braket{\phi_+^4\phi_-^2} \\
\braket{\phi_+^2\phi_-^4}\\
\braket{\phi_-^6}
\end{array}\right).
\end{equation}
There are two independent solutions
\begin{equation}
\left(\begin{array}{l}
\braket{\phi_+^6} \\
\braket{\phi_+^4\phi_-^2} \\
\braket{\phi_+^2\phi_-^4}\\
\braket{\phi_-^6}
\end{array}\right)
=\left(\begin{array}{r}1 \\-2/5 \\3/5 \\0\end{array}\right)\braket{\phi_+^6}
	+\left(\begin{array}{r}0 \\3/5 \\-2/5 \\1\end{array}\right)\braket{\phi_-^6}.
\end{equation}
\begin{eqnarray}
\fl
\therefore
\braket{\phi(\theta)^6}
&=\cos(\theta)^6\braket{\phi_+^6}
	+15\cos(\theta)^4\sin(\theta)^2\braket{\phi_+^4\phi_-^2}
	+15\cos(\theta)^2\sin(\theta)^4\braket{\phi_+^2\phi_-^4}
	\sin(\theta)^6\braket{\phi_-^6}\nonumber\\
\fl
&=\cos(3\theta)^2\braket{\phi_+^6}
	+\sin(3\theta)^2\braket{\phi_-^6}.
\end{eqnarray}

Now,
let us proceed to calculate 
$\sum \braket{\phi(\theta)^3\psi_\alpha}^2/\lambda_\alpha$.
\begin{eqnarray}
\fl
\mbox{For }\mathcal{P_C}\psi_\alpha=\psi_\alpha
\mbox{: }
\left(\begin{array}{r}
\braket{\phi_+^3\psi_{\alpha+}}\\
\braket{\phi_+\phi_-^2\psi_{\alpha+}}
\end{array}\right)
=\left(\begin{array}{r}1 \\-1\end{array}\right)
	\braket{\phi_+^3\psi_{\alpha+}},\\
\fl
\left(\begin{array}{r}
\braket{\phi_+^2\phi_-\psi_{\alpha-}}\\
\braket{\phi_-^3\psi_{\alpha-}}
\end{array}\right)
=
\left(\begin{array}{rr}
5/8 & -3/8 \\
-9/8 & -1/8
\end{array}\right)
\left(\begin{array}{r}
\braket{\phi_+^2\phi_-\psi_{\alpha-}}\\
\braket{\phi_-^3\psi_{\alpha-}}
\end{array}\right)
=\left(\begin{array}{c}-1 \\1\end{array}\right)
	\braket{\phi_-^3\psi_{\alpha-}}.
\end{eqnarray}
\begin{equation}
\fl
\therefore
\sum_{\mathcal{P_C}\psi_\alpha=\psi_\alpha}
	\!\!\!\!\!\!\frac{\braket{\phi(\theta)^3\psi_\alpha}^2}{\lambda_\alpha}
={\left(\sum_{\mathcal{P_C}\psi_{\alpha+}=\psi_{\alpha+}}
	\!\!\!\!\!\!\frac{\braket{\phi_+^3\psi_{\alpha+}}^2}{\lambda_\alpha}
	\right)}\cos(3\theta)^2
+{\left(\sum_{\mathcal{P_C}\psi_{\beta-}=\psi_{\beta-}}
	\!\!\!\!\!\!\frac{\braket{\phi_-^3\psi_{\beta-}}^2}{\lambda_\beta}
	\right)}\sin(3\theta)^2.
\end{equation}

\begin{eqnarray}
\fl
\mbox{For }
\mathcal{P_C}\psi_\alpha=0
\mbox{: }
\left(\begin{array}{r}
	\braket{\phi_+^3\psi_{\alpha+}} \\
	\braket{\phi_+ \phi_-^2\psi_{\alpha+}} \\
	\braket{\phi_+^2 \phi_-\psi_{\alpha-}} \\
	\braket{\phi_-^3 \psi_{\alpha-}}
\end{array}\right)
&=\frac{1}{16}
	\left(\begin{array}{rrrr}
	1 & 9 & \pm 9 & \pm 9 \\
	3 & -5 & \pm 3 & \pm 3 \\
	\pm 3 &\pm 3 & -5 & 3 \\
	\pm 9 & \pm 9 & 9 & 1
	\end{array}\right)
\left(\begin{array}{r}
	\braket{\phi_+^3\psi_{\alpha+}} \\
	\braket{\phi_+ \phi_-^2\psi_{\alpha+}} \\
	\braket{\phi_+^2 \phi_-\psi_{\alpha-}} \\
	\braket{\phi_-^3 \psi_{\alpha-}}
\end{array}\right)\nonumber\\
&=\left(\begin{array}{c}1 \\1/3 \\\pm 1/3 \\\pm 1\end{array}\right)
	\braket{\phi_+^3\psi_{\alpha+}},
\end{eqnarray}
where
$\pm$  stands for the sign
for $\psi_\alpha$ in representation VI and V respectively
\begin{equation}
\fl
\therefore
\sum_{\mathcal{P_C}\psi_{\alpha}=0}
	\frac{\braket{\phi(\theta)^3\psi_\alpha}}{\lambda_\alpha}^2
=\sum_{\mathcal{P_C}\psi_{\alpha}=0}
	\frac{1}{\lambda_\alpha}
	\left(\braket{\phi(\theta)^3\psi_{\alpha+}}^2
	+\braket{\phi(\theta)^3\psi_{\alpha-}}^2
	\right)
=\sum_{\mathcal{P_C}\psi_{\alpha+}=0}
	\frac{\braket{\phi_+^3\psi_{\alpha+}}^2}{\lambda_\alpha},
\end{equation}
which is $\theta$ independent.

Similarly, all terms in $A_6(\theta)$ contribute in the form
$A_{6+}\cos(3\theta)^2+A_{6-}\sin(3\theta)^2$.

\section{The eigenvalue of Hessian at bifurcated solution}
\label{EigenvaluesOfH}
In this section, we calculate the eigenvalue of Hessian $\mathcal{H}$
at bifurcated solutions,
\begin{eqnarray}
\mathcal{H}(q_b)=\mathcal{H}(q_o+r\phi(\theta)+r\epsilon\psi)
\end{eqnarray}
by ordinary perturbation methods
using the term $r\phi(\theta)+r\epsilon\psi$ for the perturbation term.
Here, we used abbreviated notation $\epsilon\psi$ for
\begin{equation}
\epsilon\psi
=\sum_\alpha \epsilon_\alpha\psi_\alpha.
\end{equation}
Since we are considering a bifurcated solution
$q_b=q_o+r\phi(\theta)+r\epsilon\psi$,
$\phi(\theta)$ and $\epsilon\psi$ are filtered by a projection operator
for this solution: $\mathcal{P}q_b=q_b$.

The zero order Hessian and the perturbation term are
\begin{eqnarray}
\mathcal{H}(q_o+r\phi(\theta)+r\epsilon\psi)=\mathcal{H}(q_o)+\Delta U,\\
\mathcal{H}(q_o)=-\frac{d^2}{d t^2}+\frac{\partial^2 U}{\partial q^2},\\
\Delta U=\mathcal{H}(q_o+r\phi(\theta)+r\epsilon\psi)
			-\mathcal{H}(q_o)
		=\sum_{n\ge 1}\frac{r^n}{n!}
			(\phi(\theta)+\epsilon\psi)^n\frac{\partial^{n+2}U}{\partial q^{n+2}}.
\end{eqnarray}
For arbitrary functions $f$ and $g$, 
\begin{equation}
\int d t \Delta U f g
=\sum_{n\ge 1}\frac{r^n}{n!}
	\braket{(\phi(\theta)+\epsilon\psi)^n f g}.
\end{equation}
The aim of this section is to calculate the eigenvalue to order $r^2$.
Since 
$\Delta U$
is order $r$, calculations to second order perturbation
are enough.

\subsection{Non-degenerate cases}
For $\kappa$ is not degenerate,
the ordinary perturbation method yields the eigenvalue
of Hessian,
\begin{eqnarray}
\fl
K
&=\kappa
	+\int d t\, \Delta U \phi^2
	-\sum_\alpha \frac{1}{\lambda_\alpha-\kappa}
		\left(\int d t\, \Delta U \phi \psi_\alpha\right)^2
	+O(\Delta U^3)\nonumber\\
\fl
&=\kappa
	+r\braket{(\phi+\epsilon\psi)\phi^2}
	+\frac{r^2}{2}\braket{\phi^4}
	-\sum_\alpha\frac{r^2}{\lambda_\alpha}
		\braket{\phi^2\psi_\alpha}^2
	+O(r^3),
\end{eqnarray}
where we have used $O(\Delta U)=O(r)$, $\kappa=O(r)$ and $\epsilon=O(r)$.
Using 
\begin{equation}
\epsilon_\alpha
=-\frac{1}{\lambda_\alpha}\sum_{n\ge 3}\frac{r^{n-2}}{(n-1)!}
	\braket{(\phi+\epsilon\psi)^{n-1}\psi_\alpha}\nonumber
=-\frac{r}{2\lambda_\alpha}\braket{\phi^2\psi_\alpha}+O(r^2),
\end{equation}
we get
\begin{equation}
K=\kappa
	+r\braket{\phi^3}
	+\frac{r^2}{2}\left(
		\braket{\phi^4}
		-\sum_\alpha \frac{3}{\lambda_\alpha}
			\braket{\phi^2\psi_\alpha}^2
		\right)
	+O(r^3).
\end{equation}
This is equal to $d^2 S_{LS}(r)/d r^2$ of $S_{LS}(r)$ in \eref{SLS}.

\subsection{Doubly degenerate cases}
For bifurcations V and VI,
the eigenvalue $\kappa$ is doubly degenerate.
Let $\phi_1$ and $\phi_2$ be eigenfunctions for $\kappa$,
$\phi_1$ be the function that contributes to the bifurcated function
$q_b=q_o+\phi_1+\epsilon\psi$,
and $\phi_2$ be another.
Let $\mathcal{P}$ be the projection operator
for $\mathcal{P}q_b=q_b$,
then $\mathcal{P}\phi_1=\phi_1$ and $\mathcal{P}\phi_2=0$
follows.
Here $\mathcal{P}$ is one of
$\mathcal{P_M}\mathcal{P_S}$, 
$\mathcal{P_S}$, and
$\mathcal{P_{MS}}$.
Note that
$\Delta U$ is diagonalised by $\phi_1$ and $\phi_2$:
\begin{equation}
\fl
\int d t\, \Delta U \phi_1 \phi_2
=\sum_n \frac{r^n}{n!}\braket{
	(\phi_1+\epsilon\psi)^n \phi_1\phi_2
	}
=\sum_n \frac{r^n}{n!}\braket{
	(\phi_1+\epsilon\psi)^n \phi_1\mathcal{P}\phi_2
	}
=0.
\end{equation}
Here, we used the same arguments for the theorem \ref{theoremProjectionOperator}.
Then, the perturvative calculations are similar to that of non-degenerate cases
to the second order:
\begin{eqnarray}
\fl
K_1=\kappa
	+\int d t \Delta U \phi_1^2
	-\sum_\alpha\frac{1}{\lambda_\alpha-\kappa}\left(
		\int d t\, \Delta U \phi_1\psi_\alpha 
		\right)^2
	+O(\Delta U^3),\\
\fl
K_2=\kappa
	+\int d t \Delta U \phi_2^2
	-\sum_\alpha\frac{1}{\lambda_\alpha-\kappa}\left(
		\int d t\, \Delta U \phi_2 \psi_\alpha 
		\right)^2
	+O(\Delta U^3).
\end{eqnarray}
Note that  $\psi_\alpha$ in the second term of $K_2$
satisfies $\mathcal{P}\psi_\alpha=0$
because $\mathcal{P}\phi_2=0$ and 
$\Delta U$ is invariant.

\subsubsection{For bifurcated solution in V.}
For this case, $\mathcal{P}=\mathcal{P_M}\mathcal{P_S}$.
Then $\phi_1=\phi_+$ and $\phi_2=\phi_-$,
and $\epsilon\psi=\sum_\alpha \epsilon_\alpha \psi_{\alpha+}$
where $\mathcal{S}\phi_\pm = \pm \phi_\pm$ and 
$\mathcal{P_S}\psi_{\alpha+}=\psi_{\alpha+}$.
Then,
\begin{equation}
\epsilon_\alpha
=-\frac{r}{2\lambda_\alpha}\braket{\phi_+^2\psi_{\alpha+}}+O(r^2),
\end{equation}
and
\begin{eqnarray}
K_1
&=\kappa+\int d t\, \Delta U \phi_+^2
	-\sum_\alpha \frac{1}{\lambda_\alpha-\kappa}
		\left(
			\int d t\, \Delta U \phi_+ \psi_{\alpha+}
		\right)^2\nonumber\\
&=\kappa+r\braket{\phi_+^3}
	+\frac{r^2}{2}\left(
		\braket{\phi_+^4}
		-\sum_\alpha \frac{3}{\lambda_\alpha}
			\braket{\phi_+^2\psi_{\alpha+}}^2
		\right)
	+O(r^3).
\end{eqnarray}
This is equal to $d^2 S_{LS}(r)/d r^2$ of $S_{LS}(r)$ in \eref{SLS}.
For $K_2$,
\begin{eqnarray}
\fl
\int d t\, \Delta U \phi_-^2\nonumber\\
\fl
=r\braket{\phi_+\phi_-^2}
	-\sum_\alpha \frac{r^2}{2\lambda_\alpha}
		\braket{\phi_+^2\psi_{\alpha+}}
		\braket{\phi_-^2\psi_{\alpha+}}
	+\frac{r^2}{2}\braket{\phi_+^2\phi_-^2}\nonumber\\
\fl
=-r\braket{\phi_+^3}
	-\!\!\!\sum_{\mathcal{P_C}\psi_{\alpha+}=\psi_{\alpha+}}
		 \frac{r^2}{2\lambda_\alpha}\braket{\phi_+^2\psi_{\alpha+}}^2
	+\!\!\!\sum_{\mathcal{P_C}\psi_{\beta+}=0}
		 \frac{r^2}{2\lambda_\beta}\braket{\phi_+^2\psi_{\beta+}}^2
	+\frac{r^2}{6}\braket{\phi_+^4}
	+O(r^3).
\end{eqnarray}
Here we have used the relations in \ref{secA3forV} and \ref{secA4forV}.
On the other hand,
\begin{eqnarray}
\fl
-\sum_\alpha \frac{1}{\lambda_\alpha-\kappa}
	\left(
		\int d t\, \Delta U \phi_-\psi_{\alpha-}
	\right)^2
=-\sum_\alpha \frac{r^2}{\lambda_\alpha}
	\braket{\phi_+\phi_-\psi_{\alpha-}}^2+O(r^3)\nonumber\\
=-\sum_{\mathcal{P_C}\psi_{\alpha-}=0}
		\frac{r^2}{\lambda_\alpha}\braket{\phi_+^2\psi_{\alpha+}}^2
		+O(r^3).
\end{eqnarray}
Here, we have used \eref{phiPlusPhiMinusPsiMinusV} and
\eref{phiPlusSquarePsiMinusV}.
So, we get
\begin{eqnarray}
\fl
K_2
&=\kappa-r\braket{\phi_+^3}
	-\!\!\!\!\!\sum_{\mathcal{P_C}\psi_{\alpha+}=\psi_{\alpha+}}
		 \frac{r^2}{2\lambda_\alpha}\braket{\phi_+^2\psi_{\alpha+}}^2
	-\!\!\!\!\!\sum_{\mathcal{P_C}\psi_{\beta+}=0}
		 \frac{r^2}{2\lambda_\beta}\braket{\phi_+^2\psi_{\beta+}}^2
	+\frac{r^2}{6}\braket{\phi_+^4}
	+O(r^3)\nonumber\\
\fl
&=\kappa
	-r\braket{\phi_+^3}
	-\sum_{\alpha}
		 \frac{r^2}{2\lambda_\alpha}\braket{\phi_+^2\psi_{\alpha+}}^2
	+\frac{r^2}{6}\braket{\phi_+^4}
	+O(r^3)\nonumber\\
\fl
&=\kappa-A_3(0)r+\frac{A_4(0)}{3!}r^2+O(r^3).
\end{eqnarray}
Here we have used the relations in \ref{secA3forV} and \ref{secA4forV}.
Substituting $r=r_b$ in \eref{rbV},
we get the same expression as in \eref{K2V}.

\subsubsection{For bifurcated solution in VI with $\mathcal{P}=\mathcal{P_S}$.}
For this solution,
\begin{eqnarray}
\phi_1=\phi_+,\, \phi_2=\phi_-,\,
\epsilon\psi=\sum_\alpha \epsilon_\alpha \psi_{\alpha+},\\
\mathcal{P_M}\phi_\pm=0,\\
\epsilon_\alpha=-\frac{r}{2\lambda_\alpha}\braket{\phi_+^2\psi_{\alpha+}}.
\end{eqnarray}
Since, $\braket{\phi_+^3}=0$ by $\mathcal{M}\phi_+=-\phi_+$,
\begin{eqnarray}
\fl
K_1=\kappa
	+\int d t\, \Delta U \phi_+^2
	-\sum_\alpha \frac{1}{\lambda_\alpha-\kappa}
		\left(
			\int d t\, \Delta U \phi_+\psi_{\alpha+}
		\right)^2
	+O(\Delta U^3)\nonumber\\
=\kappa + \frac{A_4(0)}{2}r^2+O(r^3).
\end{eqnarray}
This is equal to $\partial_r^2 S_{LS}(r,\theta)$ to $r^2$ in \eref{SLSforVI}.
On the other hand,
\begin{eqnarray}
K_2
=\kappa
	+\int d t\, \Delta U \phi_-^2
	-\sum_\alpha \frac{1}{\lambda_\alpha-\kappa}
		\left(
			\int d t\, \Delta U \phi_-\psi_{\alpha-}
		\right)^2
	+O(\Delta U^3)\nonumber\\
=\kappa+\frac{r^2}{6}\braket{\phi_+^4}
	-\sum_{\alpha}
		\frac{r^2}{2\lambda_\alpha}\braket{\phi_+^2\psi_{\alpha+}}^2
	+O(r^3)\nonumber\\
=\kappa+\frac{r^2}{3!}A_4(0)+O(r^3).\label{VIK2Plus}
\end{eqnarray}
Here we have used the relations in \ref{secA3forV} and \ref{secA4forV}.
Substituting $r=r_b$ in \eref{rbPlusForVI},
\begin{equation}
K_2=0+O(\kappa^2)
\end{equation}
that is equal to $r^{-2}\partial_\theta^2 S_{LS}(r,\theta)$ 
in \eref{VIpatialThetaSquareS} to $\kappa$.

\subsubsection{For bifurcated solution for VI with $\mathcal{P}=\mathcal{P_{MS}}$.}
For this solution,
\begin{eqnarray}
\phi_1=\phi_-,\, \phi_2=\phi_+,\\
\mathcal{P_M}\phi_\pm=0,\\
\Delta U
=\sum_{n\ge 1} \frac{r^n}{n!}(\phi_-+\mathcal{P_{MS}}\,\epsilon\psi)^n
			\frac{\partial^{n+2}U}{\partial q^{n+2}}.
\end{eqnarray}
For $\epsilon_\alpha=-r\braket{\phi_-^2\psi_{\alpha}}/(2\lambda_\alpha)+O(r^2)$,
only terms of $\psi_{\alpha+}$ survive by $\mathcal{S}$ symmetry,
\begin{equation}
\epsilon_\alpha=-\frac{r}{2\lambda_\alpha}\braket{\phi_-^2\psi_{\alpha+}}.
\end{equation}

\begin{eqnarray}
K_1=\kappa
	+\int d t\, \Delta U \phi_-^2
	-\sum_\alpha \frac{1}{\lambda_\alpha-\kappa}
		\left(
			\int d t\, \Delta U \phi_- \psi_\alpha
		\right)^2
	+O(\Delta U^3)\nonumber\\
=\kappa
	+\frac{1}{2}\braket{\phi_-^4}
	-\sum_\alpha \frac{3r^2}{2\lambda_\alpha}\braket{\phi_-^2\psi_+}^2
	+O(r^3)\nonumber\\
=\kappa
	+\frac{A_4(0)}{2}r^2+O(r^3).
\end{eqnarray}
This is equal to $\partial_r^2 S_{LS}(r,\theta)$ to $r^2$ in \eref{SLSforVI}.
On the other hand,
\begin{eqnarray}
K_2=\kappa
	+\int d t\, \Delta U\phi_+^2
	-\sum_\alpha \frac{1}{\lambda_\alpha-\kappa}
		\left(
			\int d t\, \Delta U \phi_+\psi_\alpha
		\right)^2
	+O(\Delta U^3)\nonumber\\
=\frac{r^2}{2}\braket{\phi_-^2\phi_+^2}
	-\frac{r^2}{2\lambda_\alpha}
		\braket{\phi_+^2\psi_{\alpha+}}
		\braket{\phi_-^2\psi_{\alpha+}}
	-\sum \frac{r^2}{\lambda_\alpha}
		\braket{\phi_-\phi_+\psi_\alpha}^2\nonumber\\
=\kappa+\frac{r^2}{3!}A_4(0)+O(r^3)
\end{eqnarray}
This is the same expression for $K_2$ in \eref{VIK2Plus} to $r^2$ term,
and is equal to $r^{-2}\partial_\theta^2 S_{LS}(r,\theta)$ 
in \eref{VIpatialThetaSquareS} to $\kappa$.

\vspace{2pc}
\noindent{\small
{\bf ORCID iDs}\\
Toshiaki Fujiwara: https://orcid.org/0000-0002-6396-3037\\
Hiroshi Fukuda: https://orcid.org/0000-0003-4682-9482\\
Hiroshi Ozaki: https://orcid.org/0000-0002-8744-3968\\
}

\section*{References}
\bibliographystyle{plain}
\bibliography{figureEight2019.bib}

\end{document}